\documentclass[runningheads]{svjour3}
\smartqed

\usepackage{graphicx}
\usepackage{mathptmx}
\usepackage{amssymb}
\usepackage[all]{xy}
\usepackage{lineno}
\usepackage{eucal}
\usepackage{upgreek}
\usepackage{amsmath}
\usepackage{marvosym}

\journalname{Funct. Anal. Other Math. (210)}

\spnewtheorem{primer}{Example}{\bfseries}{\rmfamily}
\spnewtheorem{zamechanie}{Remark}{\bfseries}{\rmfamily}
\spnewtheorem{gipoteza}{Conjecture}{\bfseries}{\itshape}

\renewcommand*{\leq}{\leqslant} \renewcommand*{\geq}{\geqslant}

\newcommand*{\st}{^\ast}
\newcommand*{\dtn}{d\times n}

\newcommand*{\mE}{\mathbb E} \newcommand*{\mN}{\mathbb N}
\newcommand*{\mR}{\mathbb R} \newcommand*{\mS}{\mathbb S}

\newcommand*{\ba}{\mathbf a} \newcommand*{\bb}{\mathbf b}
\newcommand*{\be}{\mathbf e} \newcommand*{\bg}{\mathbf g}
\newcommand*{\bq}{\mathbf q} \newcommand*{\br}{\mathbf r}
\newcommand*{\bs}{\mathbf s} \newcommand*{\bu}{\mathbf u}
\newcommand*{\bv}{\mathbf v} \newcommand*{\bw}{\mathbf w}
\newcommand*{\bx}{\mathbf x} \newcommand*{\by}{\mathbf y}

\newcommand*{\fa}{\mathfrak a} \newcommand*{\fb}{\mathfrak b}
\newcommand*{\fd}{\mathfrak d} \newcommand*{\fm}{\mathfrak m}
\newcommand*{\fn}{\mathfrak n} \newcommand*{\fv}{\mathfrak v}
\newcommand*{\fF}{\mathfrak F} \newcommand*{\fL}{\mathfrak L}
\newcommand*{\fP}{\mathfrak P} \newcommand*{\fX}{\mathfrak X}
\newcommand*{\fZ}{\mathfrak Z}
\newcommand*{\fso}{\mathfrak{so}}

\newcommand*{\eC}{\EuScript C} \newcommand*{\eK}{\EuScript K}
\newcommand*{\eM}{\EuScript M} \newcommand*{\eN}{\EuScript N}
\newcommand*{\eQ}{\EuScript Q} \newcommand*{\eR}{\EuScript R}
\newcommand*{\eU}{\EuScript U} \newcommand*{\eV}{\EuScript V}

\newcommand*{\rAr}{\mathrm{Ar}}
\newcommand*{\rF}{\mathrm F} \newcommand*{\rH}{\mathrm H}
\newcommand*{\rO}{\mathrm O}
\newcommand*{\rSO}{\mathrm{SO}}

\newcommand*{\vare}{\varepsilon} \newcommand*{\vart}{\vartheta}
\newcommand*{\vark}{\varkappa} \newcommand*{\vars}{\varsigma}
\newcommand*{\varp}{\varphi}
\newcommand*{\La}{\Lambda} \newcommand*{\Ups}{\Upsilon}

\newcommand*{\dg}{\dot{g}} \newcommand*{\dw}{\dot{w}}
\newcommand*{\dbg}{\dot{\bg}} \newcommand*{\dbq}{\dot{\bq}}
\newcommand*{\dbr}{\dot{\br}} \newcommand*{\dbu}{\dot{\bu}}
\newcommand*{\dbv}{\dot{\bv}} \newcommand*{\dbw}{\dot{\bw}}
\newcommand*{\dD}{\dot{D}} \newcommand*{\dX}{\dot{X}}
\newcommand*{\dZ}{\dot{Z}}
\newcommand*{\dfX}{\dot{\fX}}
\newcommand*{\dxi}{\dot{\xi}} \newcommand*{\drho}{\dot{\rho}}
\newcommand*{\dUps}{\dot{\Ups}}

\newcommand*{\bwcm}{\bw_{\mathrm{cm}}}
\newcommand*{\dbwcm}{\dbw_{\mathrm{cm}}}

\newcommand*{\sstar}{s_\star}
\newcommand*{\overH}{\overline{H}}
\newcommand*{\Zred}{Z^{\mathrm{reduced}}}

\newcommand*{\equal}{\big|_{\text{equal masses}}}
\newcommand*{\random}{\big|_{\text{random masses}}}

\newcommand*{\diag}{\mathop{\mathrm{diag}}\nolimits}
\newcommand*{\Tr}{\mathop{\mathrm{Tr}}\nolimits}
\newcommand*{\Var}{\mathop{\mathrm{Var}}\nolimits}

\newcommand*{\PiR}{\Pi_{\mathrm{left}}}
\newcommand*{\PiQ}{\Pi_{\mathrm{right}}}

\newcommand*{\ZI}{\dZ^I}
\newcommand*{\Zrot}{\dZ^{\mathrm{rot}}}
\newcommand*{\ZR}{\dZ^{\mathrm{left}}}
\newcommand*{\ZQ}{\dZ^{\mathrm{right}}}
\newcommand*{\Zout}{\dZ^{\mathrm{out}}}
\newcommand*{\Zin}{\dZ^{\mathrm{in}}}

\newcommand*{\Tvib}{T^{\mathrm{vib}}}
\newcommand*{\Trho}{T_\rho}
\newcommand*{\TL}{T_\La}
\newcommand*{\TI}{T^I}
\newcommand*{\Trot}{T^{\mathrm{rot}}}
\newcommand*{\Txi}{T_\xi}
\newcommand*{\Text}{T^{\mathrm{ext}}}
\newcommand*{\Tint}{T^{\mathrm{int}}}
\newcommand*{\Tres}{T^{\mathrm{res}}}
\newcommand*{\TJ}{T_J}
\newcommand*{\TK}{T_K}
\newcommand*{\Tac}{T_{\mathrm{ac}}}
\newcommand*{\Eout}{E^{\mathrm{out}}}
\newcommand*{\Ein}{E^{\mathrm{in}}}
\newcommand*{\Ec}{E^{\mathrm{coupl}}}

\newcommand*{\EoutA}{\Eout_1} \newcommand*{\EoutB}{\Eout_2}
\newcommand*{\EinA}{\Ein_1} \newcommand*{\EinB}{\Ein_2}

\newcommand*{\aTres}{\Tres_+} \newcommand*{\zTres}{\Tres_-}
\newcommand*{\aTac}{\Tac^+} \newcommand*{\zTac}{\Tac^-}
\newcommand*{\aEc}{\Ec_+} \newcommand*{\zEc}{\Ec_-}

\allowdisplaybreaks

\begin{document}

\title{Statistics of energy partitions for many-particle systems in arbitrary dimension}
\titlerunning{Statistics of energy partitions in arbitrary dimension}
\author{\mbox{Vincenzo Aquilanti} \and \mbox{Andrea Lombardi} \and \mbox{Mikhail B. Sevryuk}}
\authorrunning{V.~Aquilanti, A.~Lombardi, M.B.~Sevryuk}

\institute{Vincenzo Aquilanti \at
Dipartimento di Chimica, Universit\`a degli Studi di Perugia, Perugia (Italy) \\
\email{vincenzoaquilanti@yahoo.it}
\and
Andrea Lombardi \at
Dipartimento di Chimica, Universit\`a degli Studi di Perugia, Perugia (Italy) \\
\email{ebiu2005@gmail.com}
\and
Mikhail B. Sevryuk (\Letter) \at
V.L.~Tal$'$roze Institute of Energy Problems of Chemical Physics of the Russia Academy of Sciences, Moscow (Russia) \\
\email{sevryuk@mccme.ru}}

\date{Received: 27~November 2012 / Accepted: 4~December 2012}

\maketitle

\begin{abstract}
In some previous articles, we defined several partitions of the total kinetic energy $T$ of a system of $N$ classical particles in $\mR^d$ into components corresponding to various modes of motion. In the present paper, we propose formulas for the mean values of these components in the normalization $T=1$ (for any $d$ and $N$) under the assumption that the masses of all the particles are equal. These formulas are proven at the ``physical level'' of rigor and numerically confirmed for planar systems ($d=2$) at $3\leq N\leq 100$. The case where the masses of the particles are chosen at random is also considered. The paper complements our article of 2008 [Russian J Phys Chem~B 2(6):947--963] where similar numerical experiments were carried out for spatial systems ($d=3$) at $3\leq N\leq 100$.
\end{abstract}



\noindent\smash{\rule[12.9cm]{\textwidth}{1cm}}
\vspace*{-8mm}

\section{Introduction}\label{introduction}

The integral characteristic of motion in a system of classical particles is the total kinetic energy $T$. However, for quite different types of motion, the value of $T$ can obviously be the same. For instance, consider a system of two particles with the fixed center-of-mass and let $\br=\br(t)$ be the vector connecting the particles. Depending on the forces and initial conditions, this two-particle system can \emph{vibrate} (if the direction of $\br$ is not changing) or \emph{rotate} (if the length of $\br$ is not changing), and the kinetic energy in both the cases may attain any positive value. However, these two kinds of motion are in fact extremes: a typical motion is a mixture of vibrations and rotations. Of course, in the case of two particles, the total kinetic energy $T$ can be straightforwardly and naturally decomposed as the sum of two terms corresponding to vibrations and rotations. Indeed, let $\dbr=d\br/dt=\dbr_\parallel+\dbr_\perp$ where $\dbr_\parallel$ is the component parallel to $\br$ and $\dbr_\perp$ is the component orthogonal to $\br$. Then
\[
T = \frac{m_1m_2}{2(m_1+m_2)}|\dbr|^2 = \Tvib+\Trot,
\]
where
\[
\Tvib=\frac{m_1m_2}{2(m_1+m_2)}\bigl| \dbr_\parallel \bigr|^2=T\cos^2\theta
\]
is the vibrational energy and
\[
\Trot=\frac{m_1m_2}{2(m_1+m_2)}\bigl| \dbr_\perp \bigr|^2=T\sin^2\theta
\]
is the rotational energy (here $m_1$, $m_2$ are the masses of the particles and $\theta$ is the angle between $\br$ and $\dbr$).

In the case of three or more particles, there are much more kinds of motion: one has to distinguish rotations of the system in question (the ``\emph{cluster}'' or ``\emph{aggregate}'') as a whole, changes in the principal moments of inertia, various rearrangements of particles in the cluster, etc. In this situation, it becomes a rather non-trivial and ambiguous task to define the components of $T$ corresponding to various \emph{modes} of the motion. This problem has been discussed in the literature for decades and several approaches to kinetic energy partitioning have been proposed; see e.g.\ the well known papers \cite{Eckart,FTSmith,Chapuisat,Littlejohn}, the recent studies by Marsden and coworkers \cite{Marsden1,Marsden2}, and references therein.

The importance of exploring the contributions of various motion modes to the total kinetic energy $T$ stems, to a large extent, from the fact that such energy components may most probably be used as effective \emph{global indicators} of dynamical features and critical phenomena (e.g.\ phase transitions) in classical clusters. If these energy terms are defined in a sufficiently ``symmetric'' and ``invariant'' manner and can be computed automatically and fast from the coordinates and velocities of the particles, then one expects to be able to straightforwardly detect structural metamorphoses in the cluster by observing abrupt changes in the way the total kinetic energy is distributed among the modes. In the case of large clusters, this is crucial for applications because it is much easier to trace a few indicators than to analyze a huge collection of data pertaining to all the particles.

Starting in 2002, we have published a series of papers \cite{A1,A2,A3,A4,A} where we proposed and preliminarily tested a number of such global indicators on the basis of the so-called \emph{hyperspherical approach} to cluster dynamics. In particular, several partitions of the total kinetic energy $T$ with various amazing features were defined in~\cite{A3} for particles in the conventional three-dimensional space and in~\cite{A4} for particles in the Euclidean space of an arbitrary dimension $d$. Moreover, our long paper~\cite{A4} contains also \emph{rigorous mathematical proofs} of many properties (including the invariance under certain group actions) of the terms of these partitions. A refinement of one of the partitions was performed in~\cite{A} and four new terms were introduced. Some perspectives of the partitions in question are discussed from the general viewpoint of the methods of molecular dynamics (with the particles being atoms or ions) in the short reviews \cite{A-Rev1,A-Rev2,A-Rev3}. Applications include studies of such phenomena and processes in chemical physics as phase transitions in small neutral argon clusters $\rAr_{\fn}$ with $\fn=3,13,38,55$ \cite{A2,A-Appl1,A-Appl2}, dynamics and thermodynamics of small ionic argon clusters $\rAr_{\fn}^+$ with $\fn=3,6,9$~\cite{A-Appl3}, the prototypical exchange reaction $\rF+\rH_2\to\rH+\rH\rF$ of atomic fluorine and molecular hydrogen~\cite{A-Appl1}, and ultrafast relaxation dynamics of krypton atomic matrices doped with a nitrogen monoxide molecule \cite{A-Appl4,A-Appl5} (some of these applications are surveyed in~\cite{A-Rev3}). These studies have confirmed the usefulness of the kinetic energy partitions for examining dynamics of classical nanoaggregates.

However, neither the theory developed in the articles \cite{A2,A3,A4} nor the applications considered in \cite{A2,A-Appl1,A-Appl2,A-Appl3,A-Appl4,A-Appl5} allow one to conclude how the energy partitions defined in \cite{A3,A4} look like for ``typical'' systems (roughly speaking, for a random choice of the coordinates and velocities of the particles) or how the statistics of the terms depends on the number $N$ of the particles and on their masses. The statistics of the kinetic energy components for systems in the conventional space $\mR^3$ was studied numerically in our paper~\cite{A} in the range $3\leq N\leq 100$ for two extreme situations: in the case where all the particles have equal masses and in the case where the masses are chosen at random. The main observation of~\cite{A} is that in the situation of equal masses, the mean values of almost all the terms of the partitions in the normalization $T=1$ are \emph{very simple} (in fact, linear fractional for $N\geq 4$) functions of $N$. For some of the terms, the paper~\cite{A} proposed also generalizations of the formulas for the mean values to an arbitrary dimension $d$ of the space.

The aim of the present article is fourfold. First, we rigorously prove the properties of the new energy terms $\EoutA$, $\EoutB$, $\EinA$, $\EinB$ announced in our previous paper~\cite{A}. Second, in the situation of equal masses, we suggest formulas for the mean values of \emph{all} the energy components (except for the so-called unbounded ones) for $N$ particles in $\mR^d$ with an \emph{arbitrary} $d$: the mean values turn out to be simple rational functions of $d$, $N$, and $\min(d,N-1)$. Third, we prove these formulas at the ``physical level'' of rigor. Fourth, we examine numerically the statistics of all the energy terms, both in the situation of equal masses and in the situation of random masses, for $N$ particles \emph{on the plane} $\mR^2$ in the same range $3\leq N\leq 100$ as in~\cite{A}. In the situation of equal masses, our numerical experiments do confirm the formulas for the mean values of the energy terms.

It is worthwhile to note that two-dimensional physics is a well developed and rapidly progressing field of science, see e.g.\ the works \cite{Evidence,Shik1,Gomez,Shik2,Efthimiou,Abdalla,Geim,Novoselov,Novoselov0,Novoselov1} and references therein (by the way, the preprint~\cite{Efthimiou} contains 680 references). Of course, here the words ``two-dimensional'' can be used either in the mathematical sense or in the sense of monolayer films. In particular, the Nobel Prize in Physics of 2010 was awarded jointly to Andre K. Geim and Konstantin S. Novoselov ``for groundbreaking experiments regarding the two-dimensional material graphene'' \cite{Geim,Novoselov,Nobel}. Moreover, a paradigm was recently proposed in which the \emph{actual} space-time is a fundamentally $(1+1)$-dimensional universe but it is ``wrapped up'' in such a way that it appears higher dimensional (say, $2+1$ and $3+1$) at larger distances \cite{Stojkovic1,Stojkovic2,Stojkovic3}.

The paper is organized as follows. Some fundamental theoretical concepts related to our approach to kinetic energy partitions are recalled in Sections~\ref{invariants} and~\ref{momenta} while the partitions themselves are defined in Section~\ref{partitions}. Almost all the material of Sections~\mbox{\ref{invariants}--\ref{partitions}} is in fact contained in our articles \cite{A3,A4} and is reproduced here for the reader's convenience. The new energy terms $\EoutA$, $\EoutB$, $\EinA$, $\EinB$ of~\cite{A} are defined and explored in Section~\ref{elaboration}. The formulas for the mean values of the bounded energy components for $N$ particles of equal masses in $\mR^d$ are proposed in Section~\ref{contexts} and substantiated in Section~\ref{dokazatelstvo}. The numerical experiments for $d=2$ and their results are described in Section~\ref{numerics}. Concluding remarks follow in Section~\ref{conclusion}.

\section{Instantaneous phase-space invariants}\label{invariants}

Consider a system of $N$ classical particles in the Euclidean space $\mR^d$ ($d,N\in\mN=\{1; \, 2; \, \ldots\}$) with masses $m_1,m_2,\ldots,m_N$. Let $\br_1=\br_1(t), \, \br_2=\br_2(t), \, \ldots, \, \br_N=\br_N(t)$ be the radii vectors of these particles with respect to the origin. Introduce the notation
\[
M=\sum_{\alpha=1}^N m_\alpha, \qquad \bq_\alpha=(m_\alpha/M)^{1/2}\br_\alpha
\]
($1\leq\alpha\leq N$). In terms of the total mass $M$ of the system and the \emph{mass-scaled} radii vectors $\bq_\alpha$, the total kinetic energy $T$ and the total angular momentum $J$ of this system are expressed in an especially simple way:
\begin{gather}
T = \frac{1}{2}\sum_{\alpha=1}^N m_\alpha|\dbr_\alpha|^2 = \frac{M}{2}\sum_{\alpha=1}^N |\dbq_\alpha|^2,
\label{T} \\
J^2 = \sum_{1\leq i<j\leq d} J_{ij}^2, \qquad J_{ij} = \sum_{\alpha=1}^N m_\alpha(\br_{i;\alpha}\dbr_{j;\alpha}-\br_{j;\alpha}\dbr_{i;\alpha}) = M\sum_{\alpha=1}^N (\bq_{i;\alpha}\dbq_{j;\alpha}-\bq_{j;\alpha}\dbq_{i;\alpha})
\label{J}
\end{gather}
($1\leq i<j\leq d$). Here $\br_{i;\alpha}$ and $\bq_{i;\alpha}$ denote respectively the $i$th components of the vectors $\br_\alpha$ and $\bq_\alpha$ and, as usual, the dot over a letter means the time derivative. For $d=3$, the formulas~\eqref{T} and~\eqref{J} give the conventional total kinetic energy and total angular momentum of a system of $N$ particles. Thus, the mass-scaled radii vectors $\bq_\alpha(t)$ provide a more adequate description of the current state of a system of classical particles than the radii vectors $\br_\alpha(t)$ themselves.

\begin{definition}
The \emph{position matrix} of a system of $N$ classical particles in $\mR^d$ is the $d\times N$ matrix $Z$ whose columns are the mass-scaled radii vectors $\bq_1,\bq_2,\ldots,\bq_N$.
\end{definition}

Of course, the position matrix at any given time instant $t$ depends on the Cartesian coordinate frame chosen. The choice of another coordinate frame (with the same origin) is described by a transformation $Z(t)\rightsquigarrow RZ(t)$ with an orthogonal $d\times d$ matrix $R\in\rO(d)$. On the other hand, according to the general duality concept for the \emph{physical space} and the abstract ``\emph{kinematic space}'' \cite{A3,A4,A} (which are respectively $\mR^d$ and $\mR^N$ in our case), one can also consider transformations of the form $Z(t)\rightsquigarrow Z(t)Q\st$ with orthogonal $N\times N$ matrices $Q\in\rO(N)$ where the asterisk designates transposing. Such transformations correspond to changes in the ``\emph{type}'' of the coordinate frame. For instance, the passage from Cartesian coordinates to Jacobi or Radau--Smith coordinate frames (see e.g.\ the papers \cite{Littlejohn,A4,frames1,frames2} and references therein) is equivalent to a multiplication of the position matrix from the right by a suitable fixed orthogonal matrix~\cite{A4}.

\begin{primer}\label{reduction}
Suppose that the $N\geq 2$ particles in question move without external forces and their center-of-mass coincides with the origin. Then in any Cartesian coordinate frame
\[
\sum_{\alpha=1}^N m_\alpha^{1/2}\bq_\alpha = M^{-1/2}\sum_{\alpha=1}^N m_\alpha\br_\alpha \equiv 0.
\]
Choose an arbitrary matrix $Q\in\rO(N)$ with the entries in the last row equal to
\[
Q_{N\alpha}=(m_\alpha/M)^{1/2}, \quad 1\leq\alpha\leq N.
\]
Then the \emph{last column} of the matrix $Z(t)Q\st$ is zero at any time moment $t$ (and one may therefore regard the ``kinematic space'' to be $\mR^{N-1}$).
\end{primer}

We are interested in characterizing systems of classical particles in $\mR^d$ by various quantities that are determined, at any time instant $t_0$, by the positions $Z(t_0)$ and velocities $\dZ(t_0)$ of the particles\footnote{Whence the words ``\emph{phase-space}'' in Definition~\ref{ipsi} below.} at this time moment $t_0$ only\footnote{Whence the word ``\emph{instantaneous}'' in Definition~\ref{ipsi} below.} (rather than by the whole trajectory $\bigl\{ Z(t) \bigm| t_{\mathrm{begin}}\leq t\leq t_{\mathrm{end}} \bigr\}$) and are invariant under orthogonal coordinate transformations\footnote{Whence the word ``\emph{invariant}'' in Definition~\ref{ipsi} below.} both in the physical space and in the ``kinematic space''. Being inspired by the discussion above and, in particular, by Example~\ref{reduction}, one arrives at the following mathematical model.

Let $d,n\in\mN$. Consider the space $(\mR^{\dtn})^2$ of the \emph{pairs} of real $\dtn$ matrices. We will write down these pairs as $(Z,\dZ)$. On the space $(\mR^{\dtn})^2$, there acts the group $\rO(d)\times\rO(n)$:
\begin{equation}
(R,Q)(Z,\dZ) = (RZQ\st,R\dZ Q\st), \qquad R\in\rO(d), \quad Q\in\rO(n).
\label{action}
\end{equation}

\begin{definition}\label{ipsi}
An \emph{instantaneous phase-space invariant} of systems of classical particles in $\mR^d$ is a collection of functions
\[
f_{d,n}:\eM_{d,n}\times\mR^{\dtn}\to\mR, \qquad \eM_{d,n}\subset\mR^{\dtn}, \qquad n\geq n_0
\]
possessing the following two properties.

First, the definition domains $\eM_{d,n}\times\mR^{\dtn}$ of the functions $f_{d,n}$ and the functions $f_{d,n}$ themselves are invariant under the action~\eqref{action} of the group $\rO(d)\times\rO(n)$:
\[
Z\in\eM_{d,n} \quad\; \Longrightarrow \quad\; RZQ\st\in\eM_{d,n} \quad \& \quad f_{d,n}(RZQ\st,R\dZ Q\st)=f_{d,n}(Z,\dZ)
\]
for any $\dZ\in\mR^{\dtn}$ ($n\geq n_0$), $R\in\rO(d)$, and $Q\in\rO(n)$.

Second, the value of the function $f_{d,n}$ remains the same if one augments both the matrices $Z$ and $\dZ$ by the $(n+1)$th column equal to zero:
\[
Z\in\eM_{d,n} \quad\; \Longrightarrow \quad\; (Z\;0)\in\eM_{d,n+1} \quad \& \quad f_{d,n+1}\bigl( (Z\;0),(\dZ\;0) \bigr)=f_{d,n}(Z,\dZ)
\]
for any $\dZ\in\mR^{\dtn}$ ($n\geq n_0$).
\end{definition}

Most of the instantaneous phase-space invariants we will deal with will depend on the parameter $M$ (the total mass of the system of particles).

\begin{primer}\label{NvsNminus1}
Let the center-of-mass of a system of $N\geq 2$ particles in $\mR^d$ coincide with the origin. As was explained in Example~\ref{reduction}, one can choose a (non-Cartesian) coordinate frame $\fF$ in which the last column of the position matrix $Z(t)$ of such a system is identically zero. The first $N-1$ columns of $Z(t)$ constitute the $d\times(N-1)$ matrix $\Zred(t)$. Any instantaneous phase-space invariant $f$ of the system in question can be calculated either using the $d\times N$ position matrices in the initial Cartesian coordinate frame ($n=N$) or using the $d\times(N-1)$ \emph{reduced} position matrices $\Zred$ in the coordinate frame $\fF$ ($n=N-1$).
\end{primer}

In the sequel, it will be convenient to introduce the notation
\[
\fm=\min(d,n).
\]
As is very well known (see e.g.\ the manuals \cite{Golub,Horn,Voevodin,Watkins}), any $\dtn$ matrix $Z\in\mR^{\dtn}$ can be decomposed as the product of three matrices
\begin{equation}
Z=D\Ups X\st, \qquad D\in\rO(d), \quad X\in\rO(n),
\label{SVD}
\end{equation}
where all the entries of the $\dtn$ matrix $\Ups$ are zeroes with the possible exception of the diagonal entries:
\[
\Ups_{11}=\xi_1, \; \Ups_{22}=\xi_2, \; \ldots, \; \Ups_{\fm\fm}=\xi_{\fm}, \qquad \xi_1\geq\xi_2\geq\cdots\geq\xi_{\fm}\geq 0.
\]
The representation~\eqref{SVD} is called the \emph{singular value decomposition} (\emph{SVD}) of the matrix $Z$. The numbers $\xi_1,\xi_2,\ldots,\xi_{\fm}$ are called the \emph{singular values} of the matrix $Z$ and are determined uniquely by $Z$ although the orthogonal factors $D$ and $X$ in the equality~\eqref{SVD} are \emph{not}. If $n\geq d$ then the $d$ singular values of $Z$ are the square roots of the eigenvalues of the symmetric $d\times d$ matrix $ZZ\st$. If $n\leq d$ then the $n$ singular values of $Z$ are the square roots of the eigenvalues of the symmetric $n\times n$ matrix $Z\st Z$.

It is clear that each singular value of $Z$ is an instantaneous phase-space invariant; to be more precise, for each $1\leq i\leq d$, the $i$th singular value $\xi_i:\mR^{\dtn}\to\mR$ is an instantaneous phase-space invariant for $n\geq i$. More generally, any function of the singular values of $Z$ is an instantaneous phase-space invariant, and conversely, any instantaneous phase-space invariant independent of $\dZ$ is a function of the singular values of $Z$.

Along with the SVD of a matrix $Z$, it is often expedient to consider a \emph{signed singular value decomposition} (\emph{signed SVD}) \cite{A4,DE}. A signed SVD of a $\dtn$ matrix $Z$ is a representation~\eqref{SVD} where again all the entries of the $\dtn$ matrix $\Ups$ are zeroes with the possible exception of the diagonal entries $\Ups_{11},\Ups_{22},\ldots,\Ups_{\fm\fm}$, but it is no longer assumed that $\Ups_{11}\geq\Ups_{22}\geq\cdots\geq\Ups_{\fm\fm}\geq 0$. Of course, the numbers $|\Ups_{11}|,|\Ups_{22}|,\ldots,|\Ups_{\fm\fm}|$ constitute in this case an unordered collection of the singular values of $Z$. The differentiability properties of a signed SVD are usually \emph{better} than those of the SVD (see the papers \cite{A4,DE} and references therein).

Recall that the standard \emph{Frobenius inner product} on the space $\mR^{\dtn}$ of real $\dtn$ matrices is defined by the formula
\begin{equation}
\langle Z^a,Z^b\rangle = \Tr\bigl[ Z^a(Z^b)\st \bigr] = \sum_{\alpha=1}^n\sum_{i=1}^d Z^a_{i\alpha}Z^b_{i\alpha},
\label{Frobenius}
\end{equation}
where Tr denotes the trace of a square matrix. The corresponding matrix norm $\|{\cdot}\|$ given by
\[
\|Z\|^2 = \langle Z,Z\rangle = \Tr(ZZ\st) = \sum_{\alpha=1}^n\sum_{i=1}^d Z_{i\alpha}^2 = \sum_{\sigma=1}^{\fm} \xi_\sigma^2
\]
is called the \emph{Frobenius norm} \cite{Golub,Horn,Watkins}, the \emph{Euclidean norm} \cite{Horn,Voevodin}, the \emph{$l_2$-norm}~\cite{Horn}, the \emph{Schur norm}~\cite{Horn}, the \emph{Hilbert--Schmidt norm}~\cite{Horn}, or the \emph{spherical norm}~\cite{Voevodin}.

In the context of the pairs $(Z,\dZ)\in(\mR^{\dtn})^2$, the instantaneous phase-space invariant
\[
\rho=\|Z\|
\]
is called the \emph{hyperradius} (of the system of particles whose position matrix is $Z$).

\begin{zamechanie}\label{dot}
In the sequel, we will use the following notation. Given a pair $(Z,\dZ)\in(\mR^{\dtn})^2$ and a certain object (number, vector, etc.)\ $\fX=\fX(Z)$ dependent on $Z$, we will define $\dfX=\dfX(Z,\dZ)$ as
\[
\dfX=\left. \frac{d}{dt}\fX\bigl( \fZ(t) \bigr) \right|_{t=0},
\]
where $\fZ:(-\vare,\vare)\to\mR^{\dtn}$ ($0<\vare\ll 1$) is an arbitrary matrix-valued function such that
\[
\fZ(0)=Z, \qquad \left. \frac{d\fZ(t)}{dt} \right|_{t=0}=\dZ.
\]
In all the cases below, $\dfX$ will be well defined for all the pairs $(Z,\dZ)$ with the possible exception of matrices $Z$ lying in a set of positive codimension. For such $Z$, one will always be able to define the resulting quantities we will be interested in by continuity.
\end{zamechanie}

\section{Hyperangular momenta}\label{momenta}

Instantaneous phase-space invariants dependent on $\dZ$ are exemplified by the so-called \emph{hyperangular momenta} of a system of particles (with respect to the origin), namely, the \emph{physical angular momentum} $J$, the \emph{kinematic angular momentum} $K$ dual to $J$, the \emph{grand angular momentum} $\La$, and the \emph{singular angular momentum} $L$ \cite{A2,A3,A4,A,A-Rev2}. These non-negative quantities are defined by the formulas
\begin{gather}
J^2 = \sum_{1\leq i<j\leq d} J_{ij}^2, \qquad J_{ij} = M\sum_{\alpha=1}^n (Z_{i\alpha}\dZ_{j\alpha}-Z_{j\alpha}\dZ_{i\alpha}),
\label{JJ} \\
K^2 = \sum_{1\leq\alpha<\beta\leq n} K_{\alpha\beta}^2, \qquad K_{\alpha\beta} = M\sum_{i=1}^d (Z_{i\alpha}\dZ_{i\beta}-Z_{i\beta}\dZ_{i\alpha}),
\label{KK} \\[1ex]
\La^2 = M^2\sum_{\substack{1\leq i,j\leq d \\ 1\leq\alpha,\beta\leq n \\ i<j \;\;\text{ or }\;\; i=j \;\&\; \alpha<\beta}} (Z_{i\alpha}\dZ_{j\beta}-Z_{j\beta}\dZ_{i\alpha})^2,
\label{LaLa} \\[1ex]
L^2 = \sum_{1\leq\sigma<\tau\leq\fm} L_{\sigma\tau}^2, \qquad L_{\sigma\tau} = M(\xi_\sigma\dxi_\tau-\xi_\tau\dxi_\sigma)
\label{LL}
\end{gather}
($\xi_1,\xi_2,\ldots,\xi_{\fm}$ being the singular values of $Z$ and their derivatives $\dxi_\sigma$ being defined in Remark~\ref{dot}). Of course, the formula~\eqref{JJ} coincides with~\eqref{J} if $Z$ is the position matrix in a Cartesian coordinate frame and $n=N$. In our previous papers \cite{A2,A3,A4,A-Rev2}, we denoted the singular angular momentum by $L_\xi$.

The correct definition~\eqref{KK} of the kinematic angular momentum $K$ was first given in our article~\cite{A3} in the particular case $d=3$ although the words ``kinematic angular momentum'' themselves were used earlier \cite{A1,A2}. The grand angular momentum $\La$ was first introduced by Smith~\cite{FTSmith}, also in the particular dimension $d=3$ (and mainly for systems of three particles). The singular angular momentum $L$ was first defined in the paper~\cite{A2}, again in the particular case $d=3$ only. For an arbitrary dimension $d$ of the physical space, the hyperangular momenta $J$, $K$, $\La$, and $L$ were introduced in our article~\cite{A4}.

\begin{theorem}
All the four hyperangular momenta $J$, $K$, $\La$, and $L$ are instantaneous phase-space invariants.
\end{theorem}

\begin{theorem}\label{alt}
The hyperangular momenta $J$, $K$, and $\La$ can be alternatively computed as
\begin{gather}
J^2 = M^2\sum_{\alpha,\beta=1}^n \left[ \Gamma_{\alpha\beta}^{(1)}\Gamma_{\alpha\beta}^{(3)}-\Gamma_{\alpha\beta}^{(2)}\Gamma_{\beta\alpha}^{(2)} \right],
\label{JJJ} \\
\Gamma_{\alpha\beta}^{(1)} = \sum_{i=1}^d Z_{i\alpha}Z_{i\beta}, \qquad \Gamma_{\alpha\beta}^{(2)} = \sum_{i=1}^d Z_{i\alpha}\dZ_{i\beta}, \qquad \Gamma_{\alpha\beta}^{(3)} = \sum_{i=1}^d \dZ_{i\alpha}\dZ_{i\beta},
\nonumber \\
K^2 = M^2\sum_{i,j=1}^d \left[ \Delta_{ij}^{(1)}\Delta_{ij}^{(3)}-\Delta_{ij}^{(2)}\Delta_{ji}^{(2)} \right],
\label{KKK} \\
\Delta_{ij}^{(1)} = \sum_{\alpha=1}^n Z_{i\alpha}Z_{j\alpha}, \qquad \Delta_{ij}^{(2)} = \sum_{\alpha=1}^n Z_{i\alpha}\dZ_{j\alpha}, \qquad \Delta_{ij}^{(3)} = \sum_{\alpha=1}^n \dZ_{i\alpha}\dZ_{j\alpha},
\nonumber \\
\La^2 = M^2\Bigl( \|Z\|^2\|\dZ\|^2-\langle Z,\dZ\rangle^2 \Bigr).
\label{LaLaLa}
\end{gather}
There hold the inequalities $J^2+L^2\leq\La^2$ and $K^2+L^2\leq\La^2$.
\end{theorem}

\begin{zamechanie}\label{linear}
One of the important consequences of Theorem~\ref{alt} is as follows. The formulas \mbox{\eqref{JJ}--\eqref{LL}} imply that the numbers of components of the hyperangular momenta $J$, $K$, $\La$, and $L$ are equal respectively to $d(d-1)/2$, $n(n-1)/2$, $dn(dn-1)/2$, and $\fm(\fm-1)/2$. However, the alternative formulas~\eqref{KKK} and~\eqref{LaLaLa} show that $K$ and $\La$ can be calculated (like $J$) using no greater than $C_dn$ operations where the constant $C_d$ depends on the dimension $d$ only.
\end{zamechanie}

\section{Energy partitions}\label{partitions}

The instantaneous phase-space invariant
\begin{equation}
T=\frac{M}{2}\|\dZ\|^2
\label{TT}
\end{equation}
is called the \emph{total kinetic energy} (of the system of particles described by the position matrix $Z$ and its time derivative $\dZ$). Of course, the formula~\eqref{TT} coincides with~\eqref{T} if $\dZ$ is the time derivative of the position matrix in a Cartesian coordinate frame and $n=N$.

The main subject of this article is statistical properties of the terms of various partitions of the total kinetic energy $T$ (almost all of these partitions were introduced in our paper~\cite{A3} for $d=3$ and in the subsequent paper~\cite{A4} for an arbitrary dimension $d$ of the physical space). In order to define the partitions in question, one has to make oneself more acquainted with the action~\eqref{action} of the group $\rO(d)\times\rO(n)$~\cite{A4}.

Given a $\dtn$ matrix $Z$, consider the manifolds of all the matrices of the form
\[
\bigl\{ RZ \bigm| R\in\rSO(d) \bigr\}, \qquad \bigl\{ ZQ \bigm| Q\in\rSO(n) \bigr\}, \qquad \bigl\{ RZQ \bigm| R\in\rSO(d), \; Q\in\rSO(n) \bigr\}.
\]
Denote respectively the tangent spaces to these manifolds at point $Z$ by
\begin{align*}
\PiR(Z) &= \bigl\{ Z+\eR Z \bigm| \eR\in\fso(d) \bigr\}, \\
\PiQ(Z) &= \bigl\{ Z+Z\eQ \bigm| \eQ\in\fso(n) \bigr\}, \\
\Pi(Z) &= \bigl\{ Z+\eR Z+Z\eQ \bigm| \eR\in\fso(d), \; \eQ\in\fso(n) \bigr\}
\end{align*}
[recall that $\fso(d)$ and $\fso(n)$ are the spaces of skew-symmetric $d\times d$ and $n\times n$ matrices, respectively]. The multidimensional planes $\PiR(Z)$, $\PiQ(Z)$, and $\Pi(Z)$ are affine subspaces of the space $\mR^{\dtn}$ of $\dtn$ matrices.

\begin{theorem}\label{condition}
The intersection of the spaces $\PiR(Z)$ and $\PiQ(Z)$ consists of the only matrix $Z$:
\begin{align*}
\PiR(Z)\cap\PiQ(Z) = \{Z\} \quad\; &\Longleftrightarrow \quad\; \dim\PiR(Z)+\dim\PiQ(Z) = \dim\Pi(Z) \\
{} &\Longleftrightarrow \quad\; \Pi(Z)-Z = \bigl[ \PiR(Z)-Z \bigr]\oplus\bigl[ \PiQ(Z)-Z \bigr]
\end{align*}
if and only if all the \emph{positive} singular values of the matrix $Z$ are pairwise distinct.
\end{theorem}

Given now a $\dtn$ matrix $\dZ$, denote by $\ZR$, $\ZQ$, and $\Zrot$ the \emph{orthogonal projections} of $\dZ$ onto the spaces $\PiR(Z)$, $\PiQ(Z)$, and $\Pi(Z)$, respectively, in the sense of the Frobenius inner product~\eqref{Frobenius}. Let also $\ZI=\dZ-\Zrot$ denote the component of $\dZ$ orthogonal to $\Pi(Z)$.

Introduce the following non-negative quantities:
\begin{align*}
\TL &= \frac{\La^2}{2M\rho^2} & &\text{called the \emph{grand angular energy}}, \\
\Trho &= \frac{M\drho^2}{2} & &\text{called the \emph{hyperradial energy}}, \\
\Trot &= \frac{M}{2}\|\Zrot\|^2 & &\text{called the \emph{rotational energy}}, \\
\TI &= \frac{M}{2}\|\ZI\|^2 & &\text{called the \emph{inertial energy}}, \\
\Txi &= \frac{L^2}{2M\rho^2} & &\text{called the \emph{shape energy}},
\end{align*}
$\TL$, $\Trho$, and $\Txi$ being defined for $\rho>0$ only (concerning $\drho$, see Remark~\ref{dot}).

\begin{theorem}\label{So}
There hold the identities and inequalities
\begin{gather}
T=\TL+\Trho,
\label{Smith} \\
T=\Trot+\TI,
\label{orthogonal} \\
\Trot\leq\TL\leq T, \quad \Trho\leq\TI\leq T, \qquad \TL-\Trot=\TI-\Trho=\Txi\leq\min(\TL,\TI),
\nonumber \\
\Trho = \frac{M}{2\rho^2}\left( \sum_{\sigma=1}^{\fm} \xi_\sigma\dxi_\sigma \right)^2, \qquad \TI = \frac{M}{2}\sum_{\sigma=1}^{\fm} \dxi_\sigma^2.
\nonumber
\end{gather}
\end{theorem}

\begin{definition}
The equality $T=\TL+\Trho$ is called the \emph{Smith decomposition} of the total kinetic energy $T$. The equality $T=\Trot+\TI$ is called the \emph{orthogonal decomposition} of the total kinetic energy $T$.
\end{definition}

The Smith decomposition~\eqref{Smith} was deduced for the first time by Smith~\cite{FTSmith} in the particular case $N=d=3$.

The identity~\eqref{orthogonal} is in fact just the Pythagorean theorem: it follows immediately from the relations $\dZ=\Zrot+\ZI$ and $\langle\Zrot,\ZI\rangle=0$. The identity~\eqref{Smith} is an immediate consequence of the equality~\eqref{LaLaLa}. Indeed, since $\rho=\|Z\|$, one has $\rho\drho=\langle Z,\dZ\rangle$. Now
\[
\frac{2T}{M} = \|\dZ\|^2 = \frac{\|Z\|^2\|\dZ\|^2-\langle Z,\dZ\rangle^2}{\rho^2}+\frac{\langle Z,\dZ\rangle^2}{\rho^2} = \frac{\La^2}{M^2\rho^2}+\drho^2 = \frac{2\TL}{M}+\frac{2\Trho}{M}.
\]

Introduce now the non-negative quantities
\begin{align*}
\Text &= \frac{M}{2}\|\ZR\|^2 & &\text{called the \emph{external energy}}, \\
\Tint &= \frac{M}{2}\|\ZQ\|^2 & &\text{called the \emph{internal energy}}, \\
\TJ &= \frac{J^2}{2M\rho^2} & &\text{called the \emph{outer angular energy}}, \\
\TK &= \frac{K^2}{2M\rho^2} & &\text{called the \emph{inner angular energy}}
\end{align*}
as well as the quantities
\begin{align*}
\Tres &= \Trot-\Text-\Tint & &\text{called the \emph{residual energy}}, \\
\Tac &= \Trot-\TJ-\TK & &\text{called the \emph{angular coupling energy}},
\end{align*}
$\TJ$, $\TK$, and $\Tac$ being defined for $\rho>0$ only. The energies $\Tres$ and $\Tac$ can be negative.

\begin{definition}
The equality
\begin{equation}
\Trot=\Text+\Tint+\Tres
\label{projective}
\end{equation}
is called the \emph{projective partition} of the rotational energy $\Trot$. The equality
\begin{equation}
\Trot=\TJ+\TK+\Tac
\label{hyperspherical}
\end{equation}
is called the \emph{hyperspherical partition} of the rotational energy $\Trot$.
\end{definition}

\begin{theorem}\label{TJText}
There hold the inequalities
\begin{gather*}
\TJ+\Txi\leq\TL, \qquad \TK+\Txi\leq\TL, \\
\TJ\leq\Text\leq\Trot, \qquad \TK\leq\Tint\leq\Trot, \qquad \Tres\leq\Tac.
\end{gather*}
If $d\leq 2$ or $n=1$ then $\TJ=\Text$. Recall that the case $n=1$ corresponds to the situations where either there is only one particle \textup($N=1$\textup) or the system consists of two particles \textup($N=2$\textup) whose center-of-mass is fixed and coincides with the origin. If $d=1$ or $n\leq 2$ then $\TK=\Tint$. Recall that the case $n=2$ corresponds to the situations where either there are two particles \textup($N=2$\textup) or the system consists of three particles \textup($N=3$\textup) whose center-of-mass is fixed and coincides with the origin.
\end{theorem}

Finally, suppose that all the \emph{positive} singular values of the matrix $Z$ are pairwise distinct (besides these positive singular values, the matrix $Z$ is allowed to possess zero singular value of arbitrary multiplicity). Then, according to Theorem~\ref{condition}, the component $\Zrot\in\bigl[ \Pi(Z)-Z \bigr]$ of $\dZ$ can be uniquely decomposed as
\[
\Zrot=\Zout+\Zin, \qquad \Zout\in\bigl[ \PiR(Z)-Z \bigr], \quad \Zin\in\bigl[ \PiQ(Z)-Z \bigr].
\]
Introduce now the non-negative quantities
\begin{align*}
\Eout &= \frac{M}{2}\|\Zout\|^2 & &\text{called the \emph{tangent} (or \emph{singular}) \emph{external energy}}, \\
\Ein &= \frac{M}{2}\|\Zin\|^2 & &\text{called the \emph{tangent} (or \emph{singular}) \emph{internal energy}}
\end{align*}
as well as the quantity
\[
\Ec = \Trot-\Eout-\Ein \qquad \text{called the \emph{tangent} (or \emph{singular}, or \emph{Coriolis}) \emph{coupling energy}}
\]
which can be negative. In our previous papers \cite{A3,A4}, the quantities $\Eout$, $\Ein$, and $\Ec$ were called respectively the \emph{outer term}, the \emph{inner term}, and the \emph{coupling}. The energy $\Ec$ was called the Coriolis coupling term in~\cite{Marsden1} (see also~\cite{Marsden2}).

\begin{definition}
The equality
\begin{equation}
\Trot=\Eout+\Ein+\Ec
\label{singular}
\end{equation}
is called the \emph{singular value expansion} (or just \emph{singular expansion}) of the rotational energy $\Trot$.
\end{definition}

The singular value expansion~\eqref{singular} is defined only if all the positive singular values of the matrix $Z$ are pairwise distinct. However, this condition is met for generic matrices $Z\in\mR^{\dtn}$~\cite{A4}. Some partitions of the total kinetic energy $T$ essentially equivalent to the equality $T=\Eout+\Ein+\Ec+\TI$ have been known long ago (see e.g.\ the papers \cite{Eckart,Chapuisat}).

\begin{theorem}\label{razlozheniya}
Let $Z=D\Ups X\st$ be a \emph{signed} SVD~\eqref{SVD} of a $\dtn$ matrix $Z$. Suppose that all the three derivatives $\dD$, $\dUps$, and $\dX$ are well defined \textup(see Remark~\ref{dot}\textup). Then
\begin{gather*}
\dD\in\fso(d), \quad \dX\in\fso(n), \qquad \dZ=\dD\Ups X\st+D\dUps X\st+D\Ups\dX\st, \\
\Zrot=\dD\Ups X\st+D\Ups\dX\st, \qquad \ZI=D\dUps X\st, \\
(Z+\dD\Ups X\st)\in\PiR(Z), \qquad (Z+D\Ups\dX\st)\in\PiQ(Z).
\end{gather*}
In particular, if all the positive singular values of the matrix $Z$ are pairwise distinct then
\[
\Zout=\dD\Ups X\st, \qquad \Zin=D\Ups\dX\st.
\]
\end{theorem}

\begin{theorem}
All the $14$ energy quantities $\TL$, $\Trho$, $\Trot$, $\TI$, $\Txi$, $\Text$, $\Tint$, $\TJ$, $\TK$, $\Tres$, $\Tac$, $\Eout$, $\Ein$, and $\Ec$ introduced in this section are instantaneous phase-space invariants.
\end{theorem}

For a fixed total kinetic energy $T$, the energies $\TL$, $\Trho$, $\Trot$, $\TI$, $\Txi$, $\Text$, $\Tint$, $\TJ$, $\TK$, $\Tres$, and $\Tac$ [the terms of the Smith decomposition~\eqref{Smith}, the orthogonal decomposition~\eqref{orthogonal}, the projective partition~\eqref{projective}, and the hyperspherical partition~\eqref{hyperspherical}] are \emph{bounded} according to Theorems~\ref{So} and~\ref{TJText}:
\begin{gather*}
0\leq\TL\leq T, \qquad 0\leq\Trho\leq T, \qquad 0\leq\Trot\leq T, \qquad 0\leq\TI\leq T, \qquad 0\leq\Txi\leq T, \\
0\leq\Text\leq\Trot, \qquad 0\leq\Tint\leq\Trot, \qquad |\Tres|\leq\Trot, \\
0\leq\TJ\leq\Trot, \qquad 0\leq\TK\leq\Trot, \qquad |\Tac|\leq\Trot.
\end{gather*}
On the other hand, the terms $\Eout\geq 0$, $\Ein\geq 0$, and $\Ec\leq\Trot$ of the singular value expansion~\eqref{singular} are \emph{not}: $\Eout$, $\Ein$, and $-\Ec$ can be arbitrarily large for any fixed value $T>0$ \cite{A4,A}. That is why we call the equality~\eqref{singular} an \emph{expansion} rather than a \emph{partition} and use the letter $E$ for its terms rather than the letter $T$.

\begin{zamechanie}\label{smysl}
In the physically most important case where $Z$ is the position matrix of a system of $N$ particles (a \emph{cluster}) whose center-of-mass is fixed at the origin, the meaning of all the terms $\TL$, $\Trho$, $\Trot$, $\TI$, $\Txi$, $\Text$, $\Tint$, $\TJ$, $\TK$, $\Tres$, $\Tac$, $\Eout$, $\Ein$, and $\Ec$ is discussed in detail in our previous papers \cite{A3,A4,A-Rev2} (where also the names of these terms and the words ``orthogonal decomposition'', ``projective partition'', ``hyperspherical partition'', ``singular value expansion'' are justified). To be brief, the hyperradial energy $\Trho$ corresponds to the contribution (to the total kinetic energy $T$) of changes in the size of the cluster as a whole and the shape energy $\Txi$, to the contribution of changes in the shape of the cluster. The inertial energy $\TI=\Trho+\Txi$ describes the contribution of all the changes in the singular values $\xi_1,\xi_2,\ldots,\xi_\omega$ where $\omega=\min(d,N-1)$. The external energy $\Text$ is the contribution (to the total kinetic energy $T$) of rotations of the cluster in $\mR^d$ (about the origin) as a whole. These rotations leave the symmetric $N\times N$ matrix $Z\st Z$ invariant. The internal energy $\Tint$ is the contribution of so-called ``\emph{kinematic rotations}'' \cite{A3,A4,A,frames2} of the cluster in the ``kinematic space'' $\mR^{N-1}$ as a whole. ``Kinematic rotations'' are the cluster rearrangements that leave the symmetric $d\times d$ matrix $ZZ\st$ invariant. The rotational energy $\Trot$ corresponds to the contribution of rotations of both types, physical ones and ``kinematic'' ones, and the residual energy $\Tres$, to the coupling between these two types of rotations. Finally, the grand angular energy $\TL=\Trot+\Txi$ describes the joint contribution of conventional rotations of the cluster, ``kinematic rotations'', and ``rotations'' of the vector $(\xi_1,\xi_2,\ldots,\xi_\omega)$ of the singular values in $\mR^\omega$. The physical meaning of the terms of the hyperspherical partition~\eqref{hyperspherical} and the singular value expansion~\eqref{singular} is much less clear.
\end{zamechanie}

\section{Refinement of the singular value expansion}\label{elaboration}

The singular value expansion~\eqref{singular} of the rotational energy $\Trot$ can be partitioned further as follows~\cite{A}. Suppose that among the singular values $\xi_1,\xi_2,\ldots,\xi_{\fm}$ of the matrix $Z$, there are $k\leq\fm$ positive numbers, and these numbers are pairwise distinct (which ensures the existence of the singular value expansion):
\[
\xi_1>\xi_2>\cdots>\xi_k>\xi_{k+1}=\cdots=\xi_{\fm}=0.
\]
Then the symmetric $d\times d$ matrix $ZZ\st$ and the symmetric $n\times n$ matrix $Z\st Z$ have $k$ positive pairwise distinct eigenvalues $\xi_1^2,\xi_2^2,\ldots,\xi_k^2$, whereas the remaining eigenvalues (if any) of each of these matrices are equal to zero. We will denote the standard inner product in $\mR^d$ and $\mR^n$ by $\bx\cdot\by$ and the Euclidean norm $(\bx\cdot\bx)^{1/2}$ of a vector $\bx$ by $|\bx|$ (in fact, we already used the notation $|\bx|$ in Sections~\ref{introduction} and~\ref{invariants}).

Denote by $\bu_1,\bu_2,\ldots,\bu_k$ the unit eigenvectors of the matrix $ZZ\st$ corresponding to the eigenvalues $\xi_1^2,\xi_2^2,\ldots,\xi_k^2$, respectively. These vectors are determined unambiguously up to multiplication by $-1$. Let $U_0$ be the zero subspace (of dimension $d-k$) of the matrix $ZZ\st$, i.e., $U_0=\{\bu\in\mR^d \mid ZZ\st\bu=0\}$. Denote by $\ba_i$ the orthogonal projection of the derivative $\dbu_i$ of the vector $\bu_i$ (see Remark~\ref{dot}) onto $U_0$ in the sense of the standard inner product $\bx\cdot\by$ in $\mR^d$ ($1\leq i\leq k$). Introduce the non-negative quantities
\begin{align*}
\EoutA &= \frac{M}{2}\sum_{i=1}^k \xi_i^2\sum_{j=1}^k (\bu_j\cdot\dbu_i)^2 & &\text{called the \emph{unbounded component of the outer term}}, \\
\EoutB &= \frac{M}{2}\sum_{i=1}^k \xi_i^2|\ba_i|^2 & &\text{called the \emph{bounded component of the outer term}}.
\end{align*}
Similarly, denote by $\bv_1,\bv_2,\ldots,\bv_k$ the unit eigenvectors of the matrix $Z\st Z$ corresponding to the eigenvalues $\xi_1^2,\xi_2^2,\ldots,\xi_k^2$, respectively. These vectors are again determined unambiguously up to multiplication by $-1$. Let $V_0$ be the zero subspace (of dimension $n-k$) of the matrix $Z\st Z$, i.e., $V_0=\{\bv\in\mR^n \mid Z\st Z\bv=0\}$. Denote by $\bb_\alpha$ the orthogonal projection of the derivative $\dbv_\alpha$ of the vector $\bv_\alpha$ onto $V_0$ in the sense of the standard inner product $\bx\cdot\by$ in $\mR^n$ ($1\leq\alpha\leq k$). Introduce the non-negative quantities
\begin{align*}
\EinA &= \frac{M}{2}\sum_{\alpha=1}^k \xi_\alpha^2\sum_{\beta=1}^k (\bv_\beta\cdot\dbv_\alpha)^2 & &\text{called the \emph{unbounded component of the inner term}}, \\
\EinB &= \frac{M}{2}\sum_{\alpha=1}^k \xi_\alpha^2|\bb_\alpha|^2 & &\text{called the \emph{bounded component of the inner term}}.
\end{align*}
It is clear that the energies $\EoutA$ and $\EoutB$ do not change if one multiplies some of the vectors $\bu_1,\bu_2,\ldots,\bu_k$ by $-1$ and the energies $\EinA$ and $\EinB$ do not change if one multiplies some of the vectors $\bv_1,\bv_2,\ldots,\bv_k$ by $-1$.

\begin{zamechanie}
It is obvious that $\bu_i\cdot\dbu_j=-\bu_j\cdot\dbu_i$ and $\bu_i\cdot\dbu_i=0$ for any $1\leq i,j\leq k$. Consequently,
\[
\sum_{i=1}^k \xi_i^2\sum_{j=1}^k (\bu_j\cdot\dbu_i)^2 = \frac{1}{2}\sum_{i,j=1}^k (\xi_i^2+\xi_j^2)(\bu_i\cdot\dbu_j)^2 = \sum_{1\leq i<j\leq k} (\xi_i^2+\xi_j^2)(\bu_i\cdot\dbu_j)^2,
\]
so that
\begin{equation}
\EoutA = \frac{M}{2}\sum_{1\leq i<j\leq k} (\xi_i^2+\xi_j^2)(\bu_i\cdot\dbu_j)^2,
\label{outA-2008}
\end{equation}
and analogously
\begin{equation}
\EinA = \frac{M}{2}\sum_{1\leq\alpha<\beta\leq k} (\xi_\alpha^2+\xi_\beta^2)(\bv_\alpha\cdot\dbv_\beta)^2.
\label{inA-2008}
\end{equation}
\end{zamechanie}

For an arbitrary dimension $d$ of the physical space, the terms $\EoutA$ [in the form~\eqref{outA-2008}], $\EinA$ [in the form~\eqref{inA-2008}], $\EoutB$, and $\EinB$ were introduced in our previous paper~\cite{A}. However, in the particular case $d=3$, the energies $\EinA$ and $\EinB$ were defined in the literature earlier (using other names). For instance, the formula~(33) of the article~\cite{Marsden1} expresses the total kinetic energy $T$ of a cluster of $N\geq 5$ particles as the sum of five components. In our notation, those components are $\Eout$, $\Ec$, $\TI$, $\EinA$, and $\EinB$.

\begin{zamechanie}\label{column}
Let $Z=D\Ups X\st$ be a (signed) SVD~\eqref{SVD} of the matrix $Z$ with
\[
\Ups_{11}=\xi_1, \; \Ups_{22}=\xi_2, \; \ldots, \; \Ups_{kk}=\xi_k, \qquad \Ups_{k+1,k+1}=\cdots=\Ups_{\fm\fm}=0.
\]
Then $\bu_i$ can be chosen to be the $i$th column $D_i$ of the matrix $D$ for each $1\leq i\leq k$. Indeed, $ZZ\st=D\Ups\Ups\st D\st$. If $\be_i$ denotes the $i$th unit coordinate vector of $\mR^d$, then
\[
ZZ\st D_i=D\Ups\Ups\st D\st D_i=D\Ups\Ups\st\be_i=\xi_i^2D\be_i=\xi_i^2D_i.
\]
Similarly, $\bv_\alpha$ can be chosen to be the $\alpha$th column $X_\alpha$ of the matrix $X$ for each $1\leq\alpha\leq k$. Indeed, $Z\st Z=X\Ups\st\Ups X\st$. If $\be_\alpha$ denotes the $\alpha$th unit coordinate vector of $\mR^n$, then
\[
Z\st ZX_\alpha=X\Ups\st\Ups X\st X_\alpha=X\Ups\st\Ups\be_\alpha=\xi_\alpha^2X\be_\alpha=\xi_\alpha^2X_\alpha.
\]
\end{zamechanie}

\begin{theorem}\label{outin}
All the four additional energy quantities $\EoutA$, $\EoutB$, $\EinA$, and $\EinB$ are instantaneous phase-space invariants, and there hold the decompositions
\begin{equation}
\Eout=\EoutA+\EoutB, \qquad \Ein=\EinA+\EinB.
\label{further}
\end{equation}
Moreover, the inequalities
\[
\EoutB\leq\Trot, \qquad \EinB\leq\Trot
\]
are valid.
\end{theorem}

In our previous paper~\cite{A}, this theorem was just announced, so we will prove it here.

\begin{proof}
Examine the unbounded and bounded components $\EinA$ and $\EinB$ of the inner term. Consider a transformation $Z'=RZQ\st$, $\dZ'=R\dZ Q\st$ with $R\in\rO(d)$, $Q\in\rO(n)$. Then $(Z')\st Z'=QZ\st ZQ\st$ and
\begin{gather*}
V'_0=QV_0, \\
\bv'_\alpha=Q\bv_\alpha, \qquad \dbv'_\alpha=Q\dbv_\alpha, \qquad \bb'_\alpha=Q\bb_\alpha \qquad (1\leq\alpha\leq k).
\end{gather*}
Consequently, the quantities $\EinA$ and $\EinB$ remain unchanged after such a transformation. If one augments both the matrices $Z$ and $\dZ$ by the $(n+1)$th column equal to zero: $Z'=(Z\;0)$, $\dZ'=(\dZ\;0)$, then
\[
\mR^{(n+1)\times(n+1)} \ni (Z')\st Z' = \begin{pmatrix} Z\st Z & 0 \\ 0 & 0 \end{pmatrix}.
\]
Thus,
\begin{gather*}
V'_0 = \bigl\{ (\bv\:w)\in\mR^{n+1} \bigm| \bv\in V_0, \; w\in\mR \bigr\}, \\
\bv'_\alpha=(\bv_\alpha\:0), \qquad \dbv'_\alpha=(\dbv_\alpha\:0), \qquad \bb'_\alpha=(\bb_\alpha\:0) \qquad (1\leq\alpha\leq k),
\end{gather*}
and the quantities $\EinA$ and $\EinB$ again remain unchanged. Thus, these quantities are instantaneous phase-space invariants.

The relations $\EinA+\EinB=\Ein$ and $\EinB\leq\Trot$ will be verified only in the case where the pair $(Z,\dZ)$ admits a signed SVD $Z=D\Ups X\st$ of the matrix $Z$ with $\Ups$ as in Remark~\ref{column} and with the \emph{well defined} derivatives $\dD$, $\dUps$, and $\dX$ (see Remark~\ref{dot}). The opposite case is highly degenerate and of no practical importance; in fact, the equality $\EinA+\EinB=\Ein$ and the inequality $\EinB\leq\Trot$ will follow for the exceptional pairs $(Z,\dZ)$ by continuity.

According to Remark~\ref{column}, one can set $\bv_\alpha=X_\alpha$ and therefore $\dbv_\alpha=\dX_\alpha$ ($1\leq\alpha\leq k$), where $X_\alpha$ denotes the $\alpha$th column of the matrix $X$ and $\dX_\alpha$ denotes the $\alpha$th column of the matrix $\dX$. According to Theorem~\ref{razlozheniya}, $\Zin=D\Ups\dX\st$ and $\Zrot=\dD\Ups X\st+D\Ups\dX\st$. Since $\EinA$, $\EinB$, $\Ein$, and $\Trot$ are instantaneous phase-space invariants, we can assume without loss of generality that $D$ is the identity $d\times d$ matrix and $X$ is the identity $n\times n$ matrix. Then
\[
\Zin=\Ups\dX\st, \qquad \Zrot=\dD\Ups+\Ups\dX\st,
\]
$\bv_\alpha=\be_\alpha$ is the $\alpha$th unit coordinate vector of $\mR^n$,
\[
Z\st Z=\diag\bigl( \xi_1^2,\xi_2^2,\ldots,\xi_k^2, \; \underbrace{0,0,\ldots,0}_{n-k} \; \bigr),
\]
and $V_0$ is the subspace of all the vectors $\bv$ of the form
\[
\bv=\bigl( \; \underbrace{0,0,\ldots,0}_k \; ,w_1,w_2,\ldots,w_{n-k} \bigr).
\]
Thus,
\[
\bv_\beta\cdot\dbv_\alpha=\dX_{\beta\alpha}, \qquad \bb_\alpha=\bigl( \; \underbrace{0,0,\ldots,0}_k \; ,\dX_{k+1,\alpha},\dX_{k+2,\alpha},\ldots,\dX_{n\alpha} \bigr)
\]
($1\leq\alpha,\beta\leq k$). Consequently,
\begin{equation}
\EinA = \frac{M}{2}\sum_{\alpha=1}^k \xi_\alpha^2\sum_{\beta=1}^k \dX_{\beta\alpha}^2, \qquad \EinB = \frac{M}{2}\sum_{\alpha=1}^k \xi_\alpha^2\sum_{\beta=k+1}^n \dX_{\beta\alpha}^2,
\label{AB}
\end{equation}
and
\[
\EinA+\EinB = \frac{M}{2}\sum_{\alpha=1}^k \xi_\alpha^2\sum_{\beta=1}^n \dX_{\beta\alpha}^2 = \frac{M}{2}\|\Ups\dX\st\|^2 = \frac{M}{2}\|\Zin\|^2 = \Ein.
\]

Moreover, the last $n-k$ columns of the $\dtn$ matrix $\dD\Ups$ vanish. Therefore,
\[
\Trot = \frac{M}{2}\|\Zrot\|^2 = \frac{M}{2}\|\dD\Ups+\Ups\dX\st\|^2
\]
is no less than $M/2$ multiplied by the sum of the squares of all the entries of the last $n-k$ columns of the $\dtn$ matrix $\Ups\dX\st$, i.e., is no less than $\EinB$ [see~\eqref{AB}].

The unbounded and bounded components $\EoutA$ and $\EoutB$ of the outer term can be treated the same way. Note only that if one augments both the matrices $Z$ and $\dZ$ by the $(n+1)$th column equal to zero: $Z'=(Z\;0)$, $\dZ'=(\dZ\;0)$, then $Z'(Z')\st=ZZ\st$, $U'_0=U_0$, and $\bu'_i=\bu_i$, $\dbu'_i=\dbu_i$, $\ba'_i=\ba_i$ for each $1\leq i\leq k$. \hfill $\square$
\end{proof}

However, it is rather difficult to reveal the physical meaning of the quantities $\EoutA$, $\EoutB$, $\EinA$, and $\EinB$. Of course, $\EoutA$ and $\EinA$ can be arbitrarily large for any fixed value $T>0$ (like $\Eout$ and $\Ein$).

\section{Equal masses and random masses}\label{contexts}

To study the statistical properties of various components (defined in Sections~\ref{partitions} and~\ref{elaboration}) of the total kinetic energy $T$ of systems of classical particles in $\mR^d$, one has to describe precisely the sampling procedure for the coordinates and velocities of the particles. We will consider systems of $N\geq 2$ particles with the center-of-mass at the origin. As in our previous paper~\cite{A}, two situations will be dealt with: \emph{particles with equal masses} and \emph{particles with random masses}.

Let $\bw_\alpha$ and $\dbw_\alpha$ ($1\leq\alpha\leq N$) be $2N$ independent random vectors, each being uniformly distributed in the unit ball in $\mR^d$ centered at the origin. Now set
\[
\bwcm=\frac{1}{N}\sum_{\alpha=1}^N \bw_\alpha, \qquad \dbwcm=\frac{1}{N}\sum_{\alpha=1}^N \dbw_\alpha
\]
(the subscript ``cm'' is for ``center-of-mass''). The masses $m_\alpha$ of the particles in the situation of \emph{equal} masses are equal to $2/N$:
\[
m_1=m_2=\cdots=m_N=\frac{2}{N}, \qquad M=\sum_{\alpha=1}^N m_\alpha = 2,
\]
and one computes
\[
\bg_\alpha=\bw_\alpha-\bwcm, \quad \dbg_\alpha=\dbw_\alpha-\dbwcm, \qquad 1\leq\alpha\leq N.
\]
In the situation of \emph{random} masses, one first chooses the masses $m_\alpha$ of the particles according to the formula
\[
m_\alpha = 2\eta_\alpha\left( \sum_{\beta=1}^N \eta_\beta \right)^{-1}, \quad 1\leq\alpha\leq N, \qquad M=\sum_{\alpha=1}^N m_\alpha = 2,
\]
where $\eta_1,\eta_2,\ldots,\eta_N$ are independent random variables uniformly distributed between $0$ and $1$. The vectors $\bg_\alpha$ and $\dbg_\alpha$ are then computed as
\[
\bg_\alpha=m_\alpha^{-1/2}(\bw_\alpha-\bwcm), \quad \dbg_\alpha=m_\alpha^{-1/2}(\dbw_\alpha-\dbwcm), \qquad 1\leq\alpha\leq N.
\]
In both the situations
\begin{equation}
\sum_{\alpha=1}^N m_\alpha^{1/2}\bg_\alpha = \sum_{\alpha=1}^N m_\alpha^{1/2}\dbg_\alpha = 0.
\label{sumgdg}
\end{equation}
Finally, the mass-scaled radii vectors $\bq_\alpha$ and their time derivatives $\dbq_\alpha$ are calculated in both the situations as
\begin{equation}
\bq_\alpha=\eC_1\bg_\alpha, \quad \dbq_\alpha=\eC_2\dbg_\alpha, \qquad 1\leq\alpha\leq N,
\label{eC12}
\end{equation}
where the positive factors $\eC_1$ and $\eC_2$ are determined from the condition
\begin{equation}
\sum_{\alpha=1}^N |\bq_\alpha|^2 = \sum_{\alpha=1}^N |\dbq_\alpha|^2 = 1.
\label{norm}
\end{equation}

In both the situations
\begin{equation}
\sum_{\alpha=1}^N m_\alpha^{1/2}\bq_\alpha = \sum_{\alpha=1}^N m_\alpha^{1/2}\dbq_\alpha = 0
\label{sumqdq}
\end{equation}
according to~\eqref{sumgdg}, so that the $d\times N$ position matrix $Z$ with columns $\bq_1,\bq_2,\ldots,\bq_N$ and its time derivative $\dZ$ with columns $\dbq_1,\dbq_2,\ldots,\dbq_N$ do correspond to a system of $N$ particles with the masses $m_1,m_2,\ldots,m_N$ and with the center-of-mass at the origin (see Example~\ref{reduction}). Moreover, the normalization~\eqref{norm} is equivalent to that $\|Z\|=\|\dZ\|=1$. Together with $M=2$ this gives $\rho=T=1$.

\begin{zamechanie}
A random point $\bw$ uniformly distributed in the unit ball in $\mR^d$ centered at the origin can be generated as $\bw=\vark^{1/d}\bs$ where $\vark$ is a random variable uniformly distributed between $0$ and $1$ and $\bs$ is a random point uniformly distributed on the unit sphere $\mS^{d-1}\subset\mR^d$ centered at the origin ($\vark$ and $\bs$ being independent). There are many methods to choose $\bs$ for $d\geq 2$, see e.g.\ the papers \cite{Sphere1,Sphere2,Sphere3,Sphere4,Sphere5} and references therein. The most ``elegant'' (and probably the best known but not the fastest) algorithm is the so-called Muller or Brown--Muller procedure~\cite{Sphere2} which runs as follows. Let $\chi_1,\chi_2,\ldots,\chi_d$ be independent random variables \emph{normally} distributed with zero mean and the same standard deviation $\vars>0$. Then the point
\[
\bs = \left( \frac{\chi_1}{\fd},\frac{\chi_2}{\fd},\ldots,\frac{\chi_d}{\fd} \right), \qquad \fd^2=\sum_{i=1}^d \chi_i^2 \quad (\fd\geq 0),
\]
is uniformly distributed on $\mS^{d-1}$. Indeed, the joint probability density function
\[
\frac{1}{\vars^d(2\pi)^{d/2}}\exp\left( -\frac{\fd^2}{2\vars^2} \right)
\]
of $\chi_1,\chi_2,\ldots,\chi_d$ depends on $\fd$ only.

Of course, for $d=2$ one can just set
\[
\bs = (\cos\varp, \; \sin\varp),
\]
where $\varp$ is uniformly distributed between $0$ and $2\pi$. For $d=3$ the standard spherical coordinates in $\mR^3$ lead to the choice~\cite{A}
\[
\bs = \bigl( (1-h^2)^{1/2}\cos\varp, \; (1-h^2)^{1/2}\sin\varp, \; h \bigr),
\]
where $h$ and $\varp$ are independent random variables uniformly distributed in the intervals $-1\leq h\leq 1$ and $0\leq\varp\leq 2\pi$.
\end{zamechanie}

The following conjecture is the main statement of this paper.

\begin{gipoteza}\label{main}
Introduce the notation
\[
\nu=N-1, \qquad \omega=\min(d,\nu).
\]
For any $d\geq 1$ and $N\geq 2$, the mathematical expectations $\mE$ of the \emph{bounded} terms $\TL$, $\Trho$, $\Trot$, $\TI$, $\Txi$, $\Text$, $\Tint$, $\Tres$, $\TJ$, $\TK$, $\Tac$, $\EoutB$, and $\EinB$ of the Smith decomposition~\eqref{Smith}, the orthogonal decomposition~\eqref{orthogonal}, the projective partition~\eqref{projective}, the hyperspherical partition~\eqref{hyperspherical}, and the singular value expansion \eqref{singular},~\eqref{further} in the situation of \emph{equal} masses are given by the formulas
\begin{gather}
\mE\TL=1-\frac{1}{d\nu}, \qquad \mE\Trho=\frac{1}{d\nu}, \qquad \mE\Trot=1-\frac{\omega}{d\nu}, \qquad \mE\TI=\frac{\omega}{d\nu},
\label{odin} \\
\mE\Txi=\frac{\omega-1}{d\nu},
\label{dva} \\
\mE\Text=\frac{\omega(2d-\omega-1)}{2d\nu}, \qquad \mE\Tint=\frac{\omega(2\nu-\omega-1)}{2d\nu}, \qquad \mE\Tres=0,
\label{tri} \\
\mE\TJ=\frac{d-1}{d\nu}, \qquad \mE\TK=\frac{\nu-1}{d\nu}, \qquad \mE\Tac=1-\frac{d+\nu+\omega-2}{d\nu},
\label{chetyre} \\
\mE\EoutB=1-\frac{\omega}{d}, \qquad \mE\EinB=1-\frac{\omega}{\nu}.
\label{pyat}
\end{gather}
\end{gipoteza}

Since we use the normalization $T=1$, the formulas \mbox{\eqref{odin}--\eqref{pyat}} give in fact the mathematical expectations of the ratios $\TL/T$, $\Trho/T$, $\Trot/T$, $\TI/T$, $\Txi/T$, $\Text/T$, $\Tint/T$, $\Tres/T$, $\TJ/T$, $\TK/T$, $\Tac/T$, $\EoutB/T$, and $\EinB/T$. These formulas reflect vividly the duality of the physical space $\mR^d$ and the ``kinematic space'' $\mR^\nu$ (see Example~\ref{reduction}). The expression for each of the quantities $\mE\TL$, $\mE\Trho$, $\mE\Trot$, $\mE\TI$, $\mE\Txi$, $\mE\Tres$, and $\mE\Tac$ is symmetric with respect to $d$ and $\nu$. If one interchanges $d$ and $\nu$, the expressions for $\mE\Text$, $\mE\TJ$, and $\mE\EoutB$ turn into those for $\mE\Tint$, $\mE\TK$, and $\mE\EinB$, respectively, and vice versa.

Note that all the four formulas~\eqref{odin}, \eqref{dva}, and~\eqref{chetyre} for the mean values $\mE\TL$, $\mE\Txi$, $\mE\TJ$, and $\mE\TK$ of the energy terms $H$ corresponding to the momenta $\La$, $L$, $J$, and $K$ have the form
\[
\mE H=\frac{\gamma-1}{d\nu},
\]
where, roughly speaking, $\gamma$ is the dimension of the space where the ``rotation'' responsible for the momentum in question takes place. This space is the matrix space $\mR^{d\times\nu}$ for $\La$, the space $\bigl\{ (\xi_1,\xi_2,\ldots,\xi_\omega) \bigr\}$ of the collections of singular values for $L$, the physical space $\mR^d$ for $J$, and the ``kinematic space'' $\mR^\nu$ for $K$.

\begin{primer}\label{distwo}
Consider a system of two particles ($N=2$) of arbitrary masses $m_1$ and $m_2$ with the center-of-mass at the origin of $\mR^d$. For $N=2$ it is not hard to obtain \cite{A4,A} that
\begin{gather*}
T=\frac{m_1m_2}{2M}|\dbr|^2, \\
\Txi=\Tint=\Tres=\TK=\Tac=\Ein=\EinA=\EinB=\EoutA=\Ec\equiv 0, \\
\Trho=\TI=T\cos^2\theta, \\
\TL=\Trot=\Text=\TJ=\Eout=\EoutB=T\sin^2\theta,
\end{gather*}
where $\br$ is the vector connecting the particles and $\theta$ is the angle between $\br$ and $\dbr$ (cf.\ Section~\ref{introduction}). It is well known~\cite{A} that
\begin{equation}
\mE\cos^2\theta=1/d, \qquad \mE\sin^2\theta=(d-1)/d
\label{cossin}
\end{equation}
for the uniform distribution of the directions of the vectors $\br$ and $\dbr$ in $\mR^d$. One thus concludes that for $N=2$ and $T=1$
\begin{gather}
\mE\Txi=\mE\Tint=\mE\Tres=\mE\TK=\mE\Tac=\mE\Ein=\mE\EinA=\mE\EinB=\mE\EoutA=\mE\Ec=0,
\nonumber \\
\mE\Trho=\mE\TI=1/d,
\label{zwei} \\
\mE\TL=\mE\Trot=\mE\Text=\mE\TJ=\mE\Eout=\mE\EoutB=(d-1)/d
\label{drei}
\end{gather}
independently of the masses $m_1$ and $m_2$. This agrees with the formulas \mbox{\eqref{odin}--\eqref{pyat}} for $\omega=\nu=1$. Thus, for $N=2$ Conjecture~\ref{main} is correct for any $d$ (and, moreover, the equality of the masses of the particles is irrelevant for $N=2$).
\end{primer}

The formula~\eqref{odin} for $\mE\Trot$ and $\mE\TI$ as well as the formula~\eqref{tri} for $\mE\Text$, $\mE\Tint$, and $\mE\Tres$ were proposed in our previous paper~\cite{A}. Moreover, the paper~\cite{A} contained also the equalities \mbox{\eqref{odin}--\eqref{pyat}} for \emph{all} the terms in the \emph{particular} dimension $d=3$. The formulas \mbox{\eqref{odin}--\eqref{pyat}} for $d=3$ were confirmed in the paper~\cite{A} for $3\leq N\leq 100$ by extensive numerical simulation. Note that for $d=3$, the equality $\mE\EoutB=1-\omega/d$ reduces to
\[
\mE\EoutB = \left[ \begin{aligned} 2/3 &\;\text{ for }\; N=2, \\ 1/3 &\;\text{ for }\; N=3, \\ 0 &\;\text{ for }\; N\geq 4. \end{aligned} \right.
\]
The equality $\mE\EoutB=0$ for $d=3$ and $N\geq 4$ (independently of the distribution of the masses) is obvious [see~\eqref{da} below]. The numerical experiments of the paper~\cite{A} for $\EoutB$ were of course carried out for $N=3$ only.

\begin{gipoteza}\label{residual}
The equality $\mE\Tres=0$ holds for any \textup{``}reasonable\textup{''} distribution of the masses of the particles \textup(not only in the situation of equal masses\textup). Moreover, $\Tres>0$ with probability $1/2$.
\end{gipoteza}

Conjecture~\ref{residual} was confirmed in the numerical simulation of our previous paper~\cite{A} for $d=3$ and $3\leq N\leq 100$ in the situations of equal and random masses. Some ``theoretical'' support for this conjecture was also presented in~\cite{A}.

\section{A ``physical'' proof of Conjecture~\ref{main}}\label{dokazatelstvo}

Consider a system of $N\geq 2$ classical particles in $\mR^d$ with masses $m_1,m_2,\ldots,m_N$ and with the center-of-mass at the origin. Such a system can be described either by the $d\times N$ position matrix $Z$ (in a Cartesian coordinate frame) and its time derivative $\dZ$ with the columns subject to the identities~\eqref{sumqdq} or by $d\times(N-1)$ matrices $Z$ and $\dZ$ with \emph{arbitrary} columns (see Example~\ref{NvsNminus1}). We will follow the second approach and regard $Z$ and $\dZ$ as matrices of Frobenius norm $1$ [cf.~\eqref{norm}] with $d$ rows and $N-1=\nu$ columns. Recall that $M=2$ according to the procedures of Section~\ref{contexts}.

The following lemma is well known~\cite{A} and almost obvious.

\begin{lemma}\label{Pr}
Let a random point $\bx$ be uniformly distributed on the unit sphere $\mS^{p+q-1}\subset\mR^{p+q}$ centered at the origin \textup($p,q\in\mN$\textup), and let $\Phi$ be an arbitrary affine plane in $\mR^{p+q}$ of dimension $p$. Denote by $\Pr_\Phi^2\bx$ the square of the length of the orthogonal projection of $\bx$ onto $\Phi$. Then the mathematical expectation $\mE\Pr_\Phi^2\bx$ of $\Pr_\Phi^2\bx$ is equal to $p/(p+q)$.
\end{lemma}

\begin{proof}
Let $(x_1,x_2,\ldots,x_{p+q})$ be a Cartesian coordinate frame in $\mR^{p+q}$. Without loss of generality, one can assume that the plane $\Phi$ is given by the equations $x_{p+1}=x_{p+2}=\cdots=x_{p+q}=0$. If $\bx=(x_1,x_2,\ldots,x_{p+q})\in\mS^{p+q-1}$ then $\Pr_\Phi^2\bx=x_1^2+x_2^2+\cdots+x_p^2$. But $x_1^2+x_2^2+\cdots+x_{p+q}^2\equiv 1$ and $\mE x_1^2=\mE x_2^2=\cdots=\mE x_{p+q}^2$ by symmetry reasons. Therefore, $\mE x_1^2=\mE x_2^2=\cdots=\mE x_{p+q}^2=1/(p+q)$ and $\mE\Pr_\Phi^2\bx=p/(p+q)$. \hfill $\square$
\end{proof}

By the way, the formula~\eqref{cossin} for $\mE\cos^2\theta$ is a particular case of Lemma~\ref{Pr} for $p=1$, $q=d-1$.

The following reformulation of Lemma~\ref{Pr} is especially useful: if a random point $\bx=(x_1,x_2,\ldots,x_{p+q})$ is uniformly distributed on the unit sphere $\mS^{p+q-1}$ centered at the origin, then $\mE\bigl( x_{\iota_1}^2+x_{\iota_2}^2+\cdots+x_{\iota_p}^2 \bigr) = p/(p+q)$ for any fixed set of indices $1\leq\iota_1<\iota_2<\cdots<\iota_p\leq p+q$.

All the $\omega=\min(d,\nu)$ singular values of a generic $d\times\nu$ matrix $Z$ are positive and pairwise distinct, whence (see~\cite{A4})
\begin{gather*}
\dim\PiR(Z)=\frac{\omega(2d-\omega-1)}{2}, \qquad \dim\PiQ(Z)=\frac{\omega(2\nu-\omega-1)}{2}, \\
\dim\Pi(Z)=\dim\PiR(Z)+\dim\PiQ(Z)=d\nu-\omega
\end{gather*}
[note that $\omega(d+\nu-\omega)=\omega\max(d,\nu)=d\nu$]. Let $\dZ$ be a random matrix uniformly distributed on the unit sphere $\mS^{d\nu-1}\subset\mR^{d\times\nu}\cong\mR^{d\nu}$ centered at the origin [observe that after the identification $\mR^{d\times\nu}\cong\mR^{d\nu}$, the Frobenius inner product~\eqref{Frobenius} becomes the standard inner product]. Since $M=2$,
\[
\Text=\|\ZR\|^2, \qquad \Tint=\|\ZQ\|^2, \qquad \Trot=\|\Zrot\|^2.
\]
On the other hand, $\ZR$, $\ZQ$, and $\Zrot$ are the orthogonal projections of $\dZ$ onto the spaces $\PiR(Z)$, $\PiQ(Z)$, and $\Pi(Z)$, respectively, in the sense of the Frobenius inner product. According to Lemma~\ref{Pr},
\[
\mE\Text=\frac{\omega(2d-\omega-1)}{2d\nu}, \qquad \mE\Tint=\frac{\omega(2\nu-\omega-1)}{2d\nu}, \qquad \mE\Trot=1-\frac{\omega}{d\nu},
\]
and
\[
\mE\TI=1-\mE\Trot=\frac{\omega}{d\nu}, \qquad \mE\Tres=\mE\Trot-\mE\Text-\mE\Tint=0.
\]
We have verified the formulas~\eqref{odin} and~\eqref{tri} for $\mE\Trot$, $\mE\TI$, $\mE\Text$, $\mE\Tint$, and $\mE\Tres$ following the reasoning in our previous paper~\cite{A}.

Since $\rho=\|Z\|$, one has $\drho=\|Z\|^{-1}\langle Z,\dZ\rangle=\|\dZ\|\cos\vart$, where $\vart$ is the angle between $Z$ and $\dZ$ in the sense of the Frobenius inner product. Let us treat $Z$ and $\dZ$ as independent random matrices uniformly distributed on the unit sphere in $\mR^{d\times\nu}$ centered at the origin. For $\Trho=M\drho^2/2=\drho^2=\|\dZ\|^2\cos^2\vart=\cos^2\vart$ we therefore have
\[
\mE\Trho=\frac{1}{d\nu}
\]
[see~\eqref{cossin}; the equality $\|Z\|=1$ is in fact irrelevant here], and
\[
\mE\TL=1-\mE\Trho=1-\frac{1}{d\nu}, \qquad \mE\Txi=\mE\TI-\mE\Trho=\frac{\omega-1}{d\nu}.
\]
We have verified the formulas~\eqref{odin} and~\eqref{dva} for $\mE\TL$, $\mE\Trho$, and $\mE\Txi$.

To calculate $\mE\TJ$ and $\mE\TK$, it is expedient to use the expressions~\eqref{JJJ} and~\eqref{KKK} for $J^2$ and $K^2$, respectively. Since $M=2$ and $\rho=1$,
\[
\mE\TJ = \mE\sum_{\alpha,\beta=1}^\nu \left[ \Gamma_{\alpha\beta}^{(1)}\Gamma_{\alpha\beta}^{(3)}-\Gamma_{\alpha\beta}^{(2)}\Gamma_{\beta\alpha}^{(2)} \right],
\]
where
\[
\Gamma_{\alpha\beta}^{(1)} = \sum_{i=1}^d Z_{i\alpha}Z_{i\beta}, \qquad \Gamma_{\alpha\beta}^{(2)} = \sum_{i=1}^d Z_{i\alpha}\dZ_{i\beta}, \qquad \Gamma_{\alpha\beta}^{(3)} = \sum_{i=1}^d \dZ_{i\alpha}\dZ_{i\beta}.
\]
Let again $Z$ and $\dZ$ be independent random matrices uniformly distributed on the unit sphere in $\mR^{d\times\nu}$ centered at the origin. Then
\[
\mE\left[ \Gamma_{\alpha\beta}^{(1)}\Gamma_{\alpha\beta}^{(3)} \right] = \mE\left[ \Gamma_{\alpha\beta}^{(2)}\Gamma_{\beta\alpha}^{(2)} \right] = 0
\]
whenever $\alpha\neq\beta$. Indeed, if one changes the signs of all the entries $Z_{i\alpha}$, $1\leq i\leq d$, the sums $\Gamma_{\alpha\beta}^{(1)}$ and $\Gamma_{\alpha\beta}^{(2)}$ would change their signs while the sums $\Gamma_{\beta\alpha}^{(2)}$ and $\Gamma_{\alpha\beta}^{(3)}$ would remain the same. Consequently,
\begin{align*}
\mE\TJ &= \mE\sum_{\alpha=1}^\nu \left\{ \Gamma_{\alpha\alpha}^{(1)}\Gamma_{\alpha\alpha}^{(3)}-\left[ \Gamma_{\alpha\alpha}^{(2)} \right]^2 \right\} = \mE\sum_{\alpha=1}^\nu \left[ \Gamma_{\alpha\alpha}^{(1)}\Gamma_{\alpha\alpha}^{(3)}\sin^2\theta_\alpha \right] \\
{} &= \sum_{\alpha=1}^\nu \mE\Gamma_{\alpha\alpha}^{(1)}\mE\Gamma_{\alpha\alpha}^{(3)}\mE\sin^2\theta_\alpha = \frac{\nu}{\nu^2}\frac{d-1}{d}=\frac{d-1}{d\nu},
\end{align*}
where $\theta_\alpha$ denotes the angle between the $\alpha$th column of $Z$ and the $\alpha$th column of $\dZ$ in $\mR^d$ ($1\leq\alpha\leq\nu$). Here we have used the facts that for each $\alpha$, the three random variables $\Gamma_{\alpha\alpha}^{(1)}$ (the square of the length of the $\alpha$th column of $Z$ in $\mR^d$), $\Gamma_{\alpha\alpha}^{(3)}$ (the square of the length of the $\alpha$th column of $\dZ$ in $\mR^d$), and $\sin^2\theta_\alpha$ are independent, and according to Lemma~\ref{Pr}
\[
\mE\Gamma_{\alpha\alpha}^{(1)} = \mE\Gamma_{\alpha\alpha}^{(3)} = \frac{d}{d\nu} = \frac{1}{\nu}, \qquad \mE\sin^2\theta_\alpha = \frac{d-1}{d}
\]
[see~\eqref{cossin}]. Analogously,
\[
\mE\TK=\frac{\nu-1}{d\nu},
\]
and
\[
\mE\Tac=\mE\Trot-\mE\TJ-\mE\TK=1-\frac{d+\nu+\omega-2}{d\nu}.
\]
We have verified the formulas~\eqref{chetyre} for $\mE\TJ$, $\mE\TK$, and $\mE\Tac$.

Finally, consider $\EoutB$ and $\EinB$. Generically, in the notation of Section~\ref{elaboration},
\begin{equation}
\begin{gathered}
k=d, \quad U_0=0, \quad \EoutB=0, \quad \EoutA=\Eout \qquad \text{for} \quad d\leq\nu, \\
k=\nu, \quad V_0=0, \quad \EinB=0, \quad \EinA=\Ein \qquad \text{for} \quad \nu\leq d.
\end{gathered}
\label{da}
\end{equation}
Let $\nu>d$ and find $\mE\EinB$. Assume all the $d$ singular values $\xi_1,\xi_2,\ldots,\xi_d$ of the matrix $Z$ to be positive and pairwise distinct (a generic setup). Then in the SVD $Z=D\Ups X\st$ of the matrix $Z$, all the three derivatives $\dD$, $\dUps$, and $\dX$ are well defined, and
\begin{equation}
\dZ=\dD\Ups X\st+D\dUps X\st+D\Ups\dX\st
\label{for1dZ}
\end{equation}
according to Theorem~\ref{razlozheniya}. Since $Z$ and $\dZ$ are independent and $\EinB$ is an instantaneous phase-space invariant (Theorem~\ref{outin}), one may suppose that $D$ is the identity $d\times d$ matrix, $X$ is the identity $\nu\times\nu$ matrix, while $\dZ$ is still a random matrix uniformly distributed on the unit sphere in $\mR^{d\times\nu}$ centered at the origin. The equality~\eqref{for1dZ} takes the form
\begin{equation}
\dZ=\dD\Ups+\dUps+\Ups\dX\st.
\label{for2dZ}
\end{equation}
The last $\nu-d$ columns of the matrices $\dD\Ups$ and $\dUps$ vanish. Consequently, the equality~\eqref{for2dZ} implies
\[
\dZ_{i\alpha}=\xi_i\dX_{\alpha i}
\]
for any $d+1\leq\alpha\leq\nu$ and $1\leq i\leq d$. Taking into account that $M=2$, the second formula~\eqref{AB} [valid for identity matrices $D$ and $X$] becomes
\[
\EinB = \sum_{i=1}^d \xi_i^2\sum_{\alpha=d+1}^\nu \dX_{\alpha i}^2 = \sum_{i=1}^d \sum_{\alpha=d+1}^\nu \dZ_{i\alpha}^2.
\]
According to Lemma~\ref{Pr},
\begin{equation}
\mE\EinB=\frac{d(\nu-d)}{d\nu}=1-\frac{d}{\nu}.
\label{votono}
\end{equation}
Now observe that the formula~\eqref{votono} for $\nu>d$ and the equality $\mE\EinB=0$ for $\nu\leq d$ can be combined into the unified formula
\[
\mE\EinB=1-\frac{\omega}{\nu}
\]
valid for any $d$ and $\nu$. Analogously,
\[
\mE\EoutB=1-\frac{\omega}{d}
\]
for any $d$ and $\nu$. We have verified the formulas~\eqref{pyat} for $\mE\EoutB$ and $\mE\EinB$. \hfill $\square$

\medskip

Unfortunately, all these arguments are not mathematically rigorous because even in the situation of equal masses, the procedure of generating the matrices $Z$ and $\dZ$ described in Section~\ref{contexts} does \emph{not} give (after the passage from $\mR^{d\times N}$ to $\mR^{d\times(N-1)}=\mR^{d\times\nu}$) matrices uniformly distributed on the unit sphere in $\mR^{d\times\nu}$ centered at the origin. This can be easily shown even in the simplest case $d=1$, $N=3$. For these values of $d$ and $N$, the procedure of Section~\ref{contexts} for equal masses (in the case of $\dZ$ for definiteness) runs as follows.

Let $\dw_1=a$, $\dw_2=b$, $\dw_3=c$ be three independent random variables uniformly distributed between $-1$ and $1$ (in the ``unit segment'' in $\mR$). One computes the numbers
\[
\dg_1=\frac{2a-b-c}{3}, \qquad \dg_2=\frac{2b-a-c}{3}, \qquad \dg_3=\frac{2c-a-b}{3}
\]
and the vector
\[
\eC(2a-b-c, \; 2b-a-c, \; 2c-a-b)
\]
of length $1$ [in the notation of~\eqref{eC12}, $\eC=\eC_2/3>0$]. Now one has to choose an arbitrary matrix
\[
Q = \begin{pmatrix} Q_{11} & Q_{12} & Q_{13} \\ Q_{21} & Q_{22} & Q_{23} \\ 3^{-1/2} & 3^{-1/2} & 3^{-1/2} \end{pmatrix} \in \rO(3)
\]
(see Example~\ref{reduction}) and calculate the vector
\begin{equation}
\eC(2a-b-c, \; 2b-a-c, \; 2c-a-b)Q\st.
\label{new}
\end{equation}
The first two components of the vector~\eqref{new} are
\begin{align*}
\dZ_1 &= \eC\bigl[ (2a-b-c)Q_{11}+(2b-a-c)Q_{12}+(2c-a-b)Q_{13} \bigr], \\
\dZ_2 &= \eC\bigl[ (2a-b-c)Q_{21}+(2b-a-c)Q_{22}+(2c-a-b)Q_{23} \bigr]
\end{align*}
($\dZ_1^2+\dZ_2^2=1$), the third component is zero. Then $(\dZ_1,\dZ_2)$ is the $1\times 2$ matrix $\dZ$ one deals with.

Do there exist fixed numbers $Q_{11}$, $Q_{12}$, $Q_{13}$, $Q_{21}$, $Q_{22}$, $Q_{23}$ such that the point $(\dZ_1,\dZ_2)$ is \emph{uniformly} distributed on the unit circle $\mS^1\subset\mR^2$ centered at the origin? The answer to this question is negative. Indeed, let $\lambda>0$. The inequalities
\begin{equation}
0\leq\dZ_1\leq\lambda\dZ_2,
\label{key}
\end{equation}
that is
\begin{align*}
0 &\leq (2a-b-c)Q_{11}+(2b-a-c)Q_{12}+(2c-a-b)Q_{13} \\
{} &\leq \lambda\bigl[ (2a-b-c)Q_{21}+(2b-a-c)Q_{22}+(2c-a-b)Q_{23} \bigr], \\
-1 &\leq a \leq 1, \\
-1 &\leq b \leq 1, \\
-1 &\leq c \leq 1,
\end{align*}
determine a certain polyhedron in the Euclidean space $\mR^3$ with coordinates $a$, $b$, $c$. The coordinates of the vertices of this polyhedron and its volume $\eV$ are piecewise rational functions of $\lambda$, $Q_{11}$, $Q_{12}$, $Q_{13}$, $Q_{21}$, $Q_{22}$, $Q_{23}$ with integer coefficients (see Remark~\ref{Volume} below). Consequently, the probability $\eV/8$ of the inequalities~\eqref{key} is also a piecewise rational function of $\lambda$, $Q_{11}$, $Q_{12}$, $Q_{13}$, $Q_{21}$, $Q_{22}$, $Q_{23}$ with integer coefficients [the denominator $8$ in $\eV/8$ being the volume of the cube $\max\bigl( |a|,|b|,|c| \bigr)\leq 1$]. On the other hand, if the point $(\dZ_1,\dZ_2)$ were uniformly distributed on $\mS^1$, then the probability of~\eqref{key} would be equal to $\frac{1}{2\pi}\arctan\lambda$.

\begin{zamechanie}\label{Volume}
Consider a collection of $\eK$ affine hyperplanes
\[
\sum_{j=1}^{\eN} \fa_{ij}x_j = \fb_i, \quad 1\leq i\leq\eK, \qquad \eK\geq\eN+1,
\]
in $\mR^{\eN}=\bigl\{ (x_1,x_2,\ldots,x_{\eN}) \bigr\}$. Then the volume of any finite polytope $\fP$ bounded by these hyperplanes is a piecewise rational function of $\fa_{ij}$, $\fb_i$ with integer coefficients. Indeed, the coordinates of the vertices of this polytope are (piecewise) rational functions of $\fa_{ij}$, $\fb_i$ with integer coefficients according to Cramer's rule. Now triangulate the polytope $\fP$ into a set of $\eN$-dimensional simplices. The volume of any simplex with vertices $(\fv_{1\iota},\fv_{2\iota},\ldots,\fv_{\eN\iota})$, $1\leq\iota\leq\eN+1$, is equal to
\[
\pm\frac{1}{\eN!}\begin{vmatrix}
1 & 1 & \cdots & 1 \\
\fv_{11} & \fv_{12} & \cdots & \fv_{1,\eN+1} \\
\fv_{21} & \fv_{22} & \cdots & \fv_{2,\eN+1} \\
\vdots & \vdots & \ddots & \vdots \\
\fv_{\eN 1} & \fv_{\eN 2} & \cdots & \fv_{\eN,\eN+1}
\end{vmatrix}.
\]
Thus, the volume of $\fP$ is a polynomial in the coordinates of its vertices with coefficients which become integers after multiplication by $\eN!$. Various formulas and algorithms for computing the volumes of polytopes in Euclidean spaces of arbitrary dimensions are presented in the articles \cite{Volume1,Volume2,Volume3,Volume4,Volume5} and references therein.
\end{zamechanie}

A genuine proof of Conjecture~\ref{main} requires further studies beyond the present paper.

\section{Numerical experiments on the plane}\label{numerics}

At previous stages of the project \cite{A4,A}, we developed and implemented \verb@Fortran@ codes for computing all the hyperangular momenta $J$, $K$, $\La$, and $L$ (defined in Section~\ref{momenta}) and the energy terms $T$, $\TL$, $\Trho$, $\Trot$, $\TI$, $\Txi$, $\Text$, $\Tint$, $\Tres$, $\TJ$, $\TK$, $\Tac$, $\Eout$, $\EoutA$, $\EoutB$, $\Ein$, $\EinA$, $\EinB$, and $\Ec$ (defined in Sections~\ref{partitions} and~\ref{elaboration}) for the physically interesting dimensions $d=2$ and $d=3$ of the ambient space. The input data for these codes are the entries of the $\dtn$ matrices $Z$ and $\dZ$ and the total mass $M$. For fixed dimension $d$, the number of operations required in the calculation grows linearly with the number $N$ of particles ($N=n$ or $N=n+1$), cf.\ Remark~\ref{linear}.

Using the codes prepared, we have verified Conjecture~\ref{main} (and, to some extent, Conjecture~\ref{residual}) for systems of classical particles on the plane ($d=2$) by numerical simulations. Within each of the two situations defined in Section~\ref{contexts} (particles with equal masses and particles with random masses) and for each value of $N$ from $3$ through $100$ (i.e., for a total of $98$ values), $\fL=10^6$ systems of $N$ particles on the Euclidean plane with the center-of-mass at the origin were chosen using a random number generator according to the procedures described in Section~\ref{contexts}. A random point $\bw$ uniformly distributed in the unit disc on $\mR^2$ centered at the origin was always generated as
\[
\bw = \vark^{1/2}(\cos\varp, \; \sin\varp),
\]
where $\vark$ and $\varp$ are independent random variables uniformly distributed in the intervals $0\leq\vark\leq 1$ and $0\leq\varp\leq 2\pi$. For each system, we calculated all the energy terms $T=1$, $\TL$, $\Trho$, $\Trot$, $\TI$, $\Txi$, $\Text$, $\Tint$, $\Tres$, $\TJ$, $\TK$, $\Tac$, $\Eout$, $\EoutA$, $\EoutB$, $\Ein$, $\EinA$, $\EinB$, and $\Ec$ defined in Sections~\ref{partitions} and~\ref{elaboration}.

Similar simulations were performed in our previous paper~\cite{A} for $d=3$ with $10^5$ systems for each value of $N$ from $3$ through $100$ in each of the two situations.

As one expects, for all the systems on $\mR^2$ in both the situations we found $\TJ=\Text$ (see Theorem~\ref{TJText}) and $\EoutB=0$, i.e., $\EoutA=\Eout$ [see~\eqref{da}]. For $N=3$ and all the systems in both the situations, we also found $\TK=\Tint$ (see Theorem~\ref{TJText}) and $\EinB=0$, i.e., $\EinA=\Ein$ [see~\eqref{da}].

The energies $\Tres$, $\Tac$, and $\Ec$ can be both positive and negative~\cite{A4}. We will denote their positive and negative ``components'' as
\begin{gather*}
\aTres=\max(\Tres,0), \qquad \zTres=-\min(\Tres,0), \qquad \Tres=\aTres-\zTres, \\
\aTac=\max(\Tac,0), \qquad \zTac=-\min(\Tac,0), \qquad \Tac=\aTac-\zTac, \\
\aEc=\max(\Ec,0), \qquad \zEc=-\min(\Ec,0), \qquad \Ec=\aEc-\zEc.
\end{gather*}

As was pointed out above, the quantities $\EoutA$ and $\EinA$ (and, consequently, $\Eout$, $\Ein$, $-\Ec$, and $\zEc$) can be arbitrarily large for any fixed value $T>0$. In our simulations, for any number $N$ of particles, we encountered systems for which $\Eout\geq 5.42592\times 10^4$ (for equal masses) or $\Eout\geq 5.1008\times 10^4$ (for random masses); recall that $\Eout\equiv\EoutA$ in our calculations. The maximal values of $\Eout$ (over all $N$) we observed turned out to be $2.34534\times 10^8$ (for equal masses) and $2.07751\times 10^7$ (for random masses). Similarly, for any number $N$ of particles, there were systems for which $\EinA\geq 5.41936\times 10^4$ (for equal masses) or $\EinA\geq 5.10724\times 10^4$ (for random masses), and the maximal values of $\EinA$ over all $N$ were equal to $2.34534\times 10^8$ (for equal masses) and $2.07787\times 10^7$ (for random masses). Finally, for any number $N$ of particles, systems occurred for which $\Ec\leq -1.08453\times 10^5$ (for equal masses) or $\Ec\leq -1.0208\times 10^5$ (for random masses), and the minimal values of $\Ec$ over all $N$ were $-4.69068\times 10^8$ (for equal masses) and $-4.15538\times 10^7$ (for random masses). Moreover, in each of the two situations, the maximal values of $\Eout$, $\EinA$, $\Ein$, and $-\Ec$ were attained at the same system---for which the ``angle'' between the spaces $\PiR(Z)$ and $\PiQ(Z)$ is very small (cf.~\cite{A4}). Since the quantities $\Eout$, $\EinA$, $\Ein$, and $\zEc$ are unbounded, we did not examine their statistics.

\begin{zamechanie}
Of course, the large number of digits in the data above only reflects the particular set of numerical experiments. The same remark refers to similar data below where only several first digits are significant and informative.
\end{zamechanie}

On the other hand, the energies $\TL$, $\Trot$, $\Trho$, $\TI$, and $\Txi$ do not exceed $T=1$ while the terms $\Text\equiv\TJ$, $\Tint$, $\Tres$, $\aTres$, $\zTres$, $\TK$, $\Tac$, $\aTac$, $\zTac$, $\EinB$, and $\aEc$ do not exceed $\Trot$ according to Theorems~\ref{So}, \ref{TJText}, and~\ref{outin}. For each of these quantities, we computed the \emph{mean values}
\[
\overH = \frac{1}{\fL}\sum_{l=1}^{\fL} H_l,
\]
where $H_l$ is the value of the energy $H$ in question for the $l$th system for the given $N$ within the given situation (recall that $\fL=10^6$), and the \emph{sample variances}
\[
s^2(H) = \overline{\left( H-\overH \,\right)^2} = \overline{H^2}-\left(\, \overH \,\right)^2
\]
(see e.g.\ the manuals \cite{Cramer,Lagutin,Waerden,Wilks} and references therein). Note that the sample variance is often defined as $\sstar^2(H)=\fL\,s^2(H)/(\fL-1)$ \cite{Waerden,Wilks}. The mathematical expectations of $s^2(H)$ and $\sstar^2(H)$ are equal to $(\fL-1)(\Var H)/\fL$ and $\Var H$, respectively \cite{Cramer,Lagutin,Waerden,Wilks}, where $\Var H = \mE\bigl[ (H-\mE H)^2 \bigr] = \mE(H^2)-(\mE H)^2$ is the variance of $H$. However, for $\fL=10^6$, the difference between the ``biased'' sample variance $s^2(H)$ and the ``unbiased'' sample variance $\sstar^2(H)$ is of course negligible.

The dependences of $\overH$ on $N$ for various terms $H$ and both the mass distributions considered are presented in Figs.~\mbox{\ref{fig1}--\ref{fig4}}. Along the abscissa axis on each of these figures, the ``physical'' distance $x$ between the left end point corresponding to $N=3$ and the point corresponding to a given $N$ is proportional to $\frac{1}{2}-\frac{1}{N-1}$. In such a coordinate frame, any dependence
\[
y=\frac{a\nu+b}{c\nu}, \qquad \nu=N-1,
\]
is represented by a \emph{straight line}:
\[
x=\frac{1}{2}-\frac{1}{\nu} \quad\; \Longleftrightarrow \quad\; \nu=\frac{2}{1-2x} \quad\; \Longrightarrow \quad\; y=\frac{a\nu+b}{c\nu}=\frac{2a+b-2bx}{2c}.
\]
In Figs.~\mbox{\ref{fig2}--\ref{fig4}}, we also show the fractions of systems for which $\Tres<0$, $\Tac<0$, or $\Ec>0$. The ``oscillations'' on the corresponding dotted lines in Figs.~\ref{fig2} and~\ref{fig4} in the region of large $N$ are due to some subtle shortcomings of the graphic system we used (\verb@gnuplot 4.0@).

\begin{center}
\begin{figure}[ht]\center
\vspace{-2.7cm}
\includegraphics[width=1.17\textwidth]{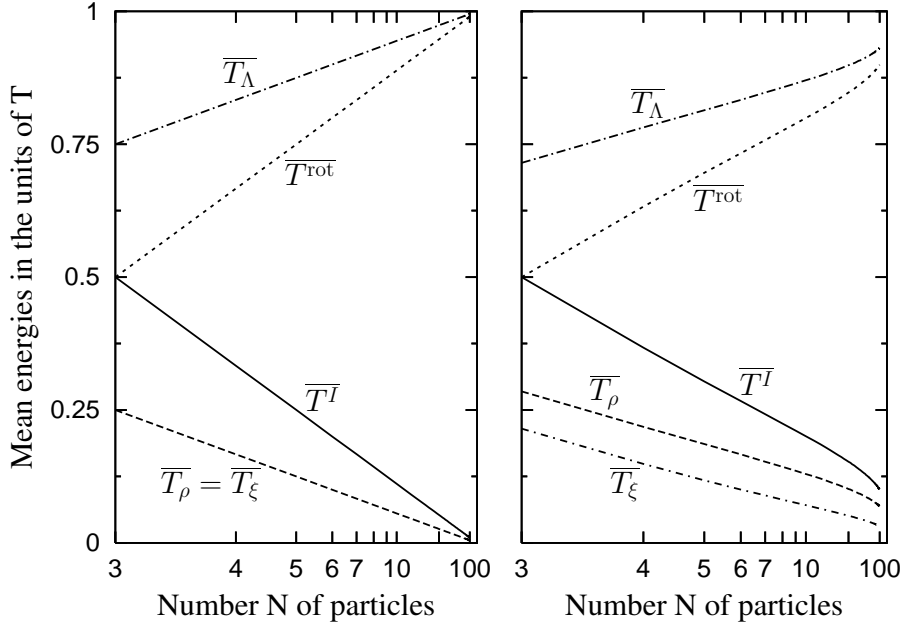}
\vspace{-9.57cm}
\caption{The mean values of the terms of the Smith decomposition~\eqref{Smith} and orthogonal decomposition~\eqref{orthogonal} of the total kinetic energy $T\equiv 1$, as well as the mean values of the shape energy $\Txi$. The left panel corresponds to the equal mass case and the right one, to the random mass case. The physical space dimension $d$ is equal to $2$.} \label{fig1}
\end{figure}
\end{center}

\begin{center}
\begin{figure}[ht]\center
\vspace{-2.7cm}
\includegraphics[width=1.17\textwidth]{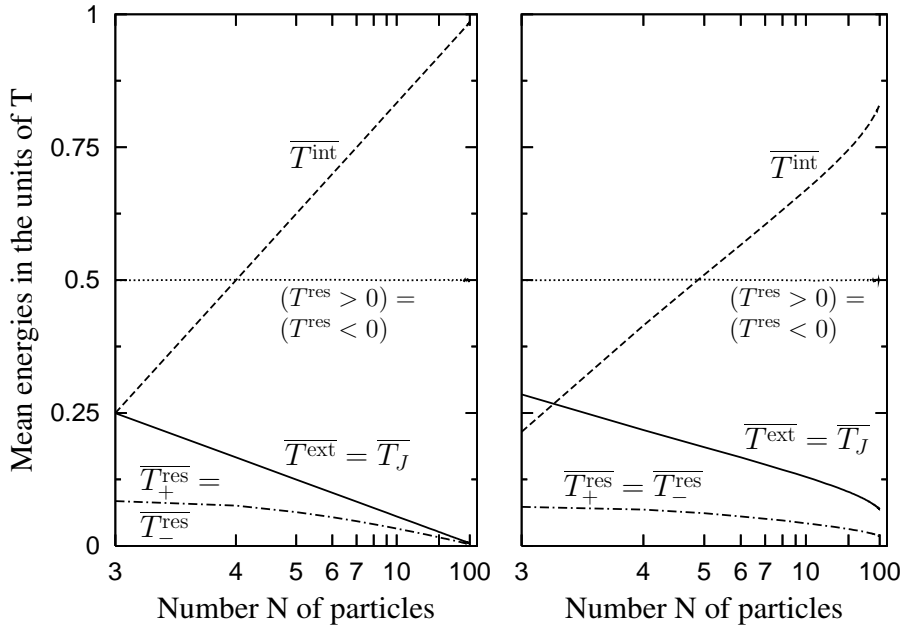}
\vspace{-9.57cm}
\caption{The mean values of the terms of the projective partition~\eqref{projective} of the rotational energy $\Trot$ for $T\equiv 1$. The left panel corresponds to the equal mass case and the right one, to the random mass case. The physical space dimension $d$ is equal to $2$, so that $\Text\equiv\TJ$ according to Theorem~\ref{TJText}. The mean values of $\Tres_+$ and $\Tres_-$ are indistinguishable on the scale of the figure, so that the mean values of the residual energy $\Tres$ itself are indistinguishable from zero. The dotted line labeled as ``$(\Tres>0)=(\Tres<0)$'' presents the fraction of systems for which $\Tres$ is positive or the fraction of systems for which $\Tres$ is negative (both these fractions are indistinguishable from $1/2$ on the scale of the figure).} \label{fig2}
\end{figure}
\end{center}

\begin{center}
\begin{figure}[ht]\center
\vspace{-2.7cm}
\includegraphics[width=1.17\textwidth]{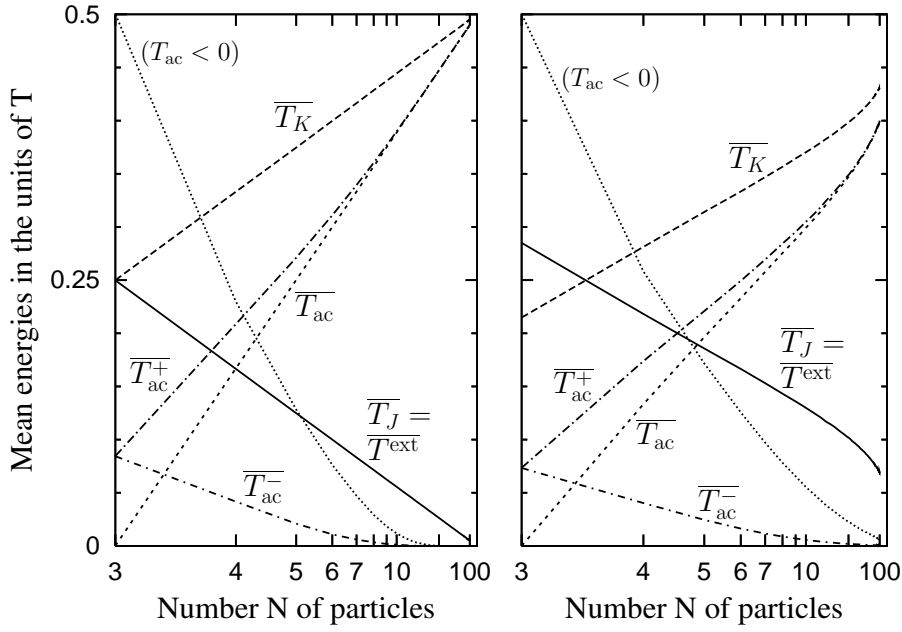}
\vspace{-9.57cm}
\caption{The mean values of the terms of the hyperspherical partition~\eqref{hyperspherical} of the rotational energy $\Trot$ for $T\equiv 1$. The left panel corresponds to the equal mass case and the right one, to the random mass case. The physical space dimension $d$ is equal to $2$, so that $\TJ\equiv\Text$ according to Theorem~\ref{TJText}. The dotted line labeled as ``$(\Tac<0)$'' presents the fraction of systems for which the angular coupling energy $\Tac$ is negative.} \label{fig3}
\end{figure}
\end{center}

\begin{center}
\begin{figure}[ht]\center
\vspace{-2.7cm}
\includegraphics[width=1.17\textwidth]{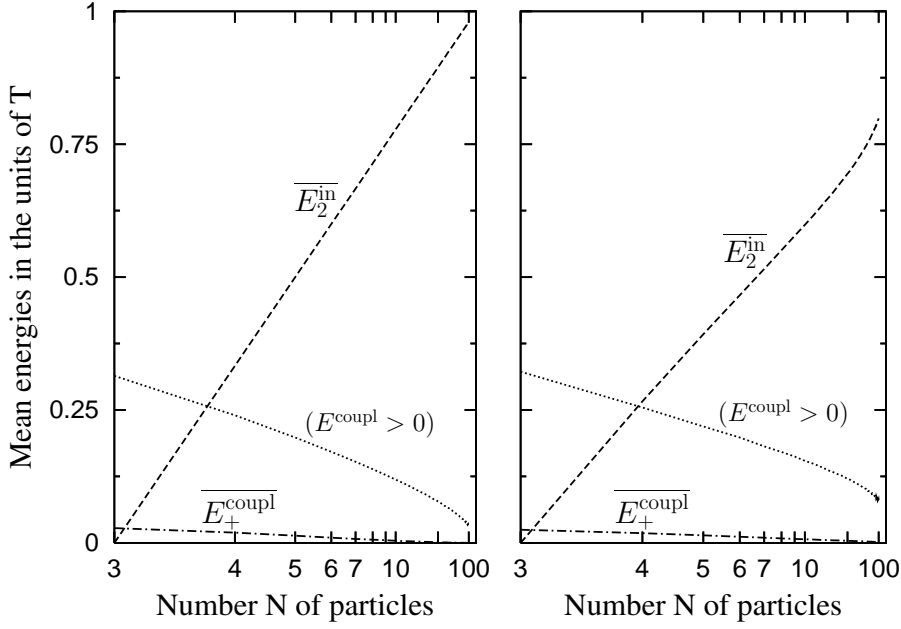}
\vspace{-9.57cm}
\caption{The mean values of the bounded terms of the singular value expansion~\eqref{singular} [and of its refinement~\eqref{further}] of the rotational energy $\Trot$ for $T\equiv 1$. The left panel corresponds to the equal mass case and the right one, to the random mass case. The physical space dimension $d$ is equal to $2$. The dotted line labeled as ``$(\Ec>0)$'' presents the fraction of systems for which the tangent coupling energy $\Ec$ is positive.} \label{fig4}
\end{figure}
\end{center}

For $d=2$ and $N\geq 3$ (and, consequently, $\nu=N-1\geq 2$ and $\omega=2$), the formulas \mbox{\eqref{odin}--\eqref{pyat}} of Conjecture~\ref{main} take the form
\begin{gather}
\mE\TL=1-\frac{1}{2\nu}, \qquad \mE\Trho=\frac{1}{2\nu}, \qquad \mE\Trot=1-\frac{1}{\nu}, \qquad \mE\TI=\frac{1}{\nu},
\label{odin2} \\
\mE\Txi=\frac{1}{2\nu},
\label{dva2} \\
\mE\Text=\frac{1}{2\nu}, \qquad \mE\Tint=1-\frac{3}{2\nu}, \qquad \mE\Tres=0,
\label{tri2} \\
\mE\TJ=\frac{1}{2\nu}, \qquad \mE\TK=\frac{\nu-1}{2\nu}, \qquad \mE\Tac=\frac{\nu-2}{2\nu},
\label{chetyre2} \\
\mE\EoutB=0, \qquad \mE\EinB=1-\frac{2}{\nu}.
\label{pyat2}
\end{gather}
Of course, the equality $\mE\Text=\mE\TJ$ follows from the fact that $\Text\equiv\TJ$ for $d=2$, and the equality $\mE\EoutB=0$ follows from the fact that generically $\EoutB=0$ for $N\geq d+1$. These two equalities hold for any distribution of the masses. Similarly, the equality $\mE\EinB=0$ for $\nu=2$ (i.e., for $N=3$) follows from the fact that generically $\EinB=0$ for $N\leq d+1$ (again for any distribution of the masses). On the other hand, the formulas \mbox{\eqref{odin2}--\eqref{tri2}} indicate that $\mE\Trho=\mE\Txi=\mE\Text$ for $d=2$ and $N\geq 3$ in the situation of equal masses (although the hyperradial energy $\Trho$, the shape energy $\Txi$, and the external energy $\Text$ of an individual system on the plane are not related to each other at all).

Thus, what we verified was the formulas \mbox{\eqref{odin2}--\eqref{chetyre2}} for the $10$ terms $\TL$, $\Trho$, $\Trot$, $\TI$, $\Txi$, $\Text$, $\Tint$, $\Tres$, $\TK$, $\Tac$ in the situation of equal masses for $3\leq N\leq 100$, the formula~\eqref{pyat2} for the term $\EinB$ in the situation of equal masses for $4\leq N\leq 100$, and the formula $\mE\Tres=0$ in the situation of \emph{random} masses for $3\leq N\leq 100$ (see Conjecture~\ref{residual})---a total of
\[
11\times 98+97=1175
\]
equalities of the form $\overH=\mE H$ where the $\mE H$ values are given by the formulas \mbox{\eqref{odin2}--\eqref{pyat2}}. All these equalities \emph{were indeed obtained} in our calculations. Namely, the minimal value, the maximal value, and the mean value of the difference $\bigl| \overH-\mE H \bigr|$ [over the $1175$ triples
\[
(\text{the term $H$}, \; N, \; \text{equal/random masses})
\]
considered] were found to be $0$,\, $4.81437\times 10^{-4}$, and $2.7892\times 10^{-5}$, respectively, and the minimal value, the maximal value, and the mean value of the ratio
\[
\frac{\bigl| \overH-\mE H \bigr|}{\bigl[ s^2(H) \big/ \fL \bigr]^{1/2}}, \qquad \fL=10^6,
\]
turned out to be equal to $0$,\, $2.41227$, and $0.698384$, respectively (so that indeed $\overH=\mE H$ up to statistical errors \cite{Cramer,Lagutin,Waerden,Wilks}, see Figs.~\mbox{\ref{fig1}--\ref{fig4}}). Note also that the minimal value, the maximal value, and the mean value of the ``weighted'' difference $2\nu\bigl| \overH-\mE H \bigr|$ were $0$,\, $2.31975\times 10^{-2}$, and $1.63148\times 10^{-3}$, respectively, while the minimal value, the maximal value, and the mean value of the sample variance $s^2(H)$ equaled $5.02904\times 10^{-5}$, $8.45973\times 10^{-2}$, and $3.5151\times 10^{-3}$, respectively.

We do not show the standard errors $\delta H=\bigl[ s^2(H) \big/ \fL \bigr]^{1/2}=10^{-3}s(H)$ of the mean values $\overH$ \cite{Cramer,Lagutin,Waerden,Wilks} in Figs.~\mbox{\ref{fig1}--\ref{fig4}} since the numbers $\overH+\delta H$ and $\overH-\delta H$ are indistinguishable on the scale of the figures.

The minimal value, the maximal value, and the mean value of the number of systems for which $\Tres<0$ (over the $98$ values of $N$) were equal to $498\,853$, $501\,308$, and $499\,930$ (of $\fL=10^6$), respectively, in the situation of equal masses and to $498\,691$, $501\,145$, and $499\,963$, respectively, in the situation of random masses. Therefore, in both the situations, $\Tres$ is negative for half of the systems up to statistical errors (see Fig.~\ref{fig2}).

We have thus confirmed Conjectures~\ref{main} and~\ref{residual} for $d=2$ and $3\leq N\leq 100$ (to be more precise, Conjecture~\ref{residual} has been confirmed only in the situations of equal and random masses). Moreover, it turns out that $s^2(\aTres)=s^2(\zTres)$ at any $N$ for both the mass distributions up to statistical errors.

The equality $\mE\Tres=0$ implies that $\mE\aTres=\mE\zTres$ for both the mass distributions. Our numerical simulations suggest that for sufficiently large $N$ (say, for $N\gtrsim 30$) in the situation of equal masses,
\[
\mE\aTres=\mE\zTres\approx\frac{5}{16\nu},
\]
and this approximation is quite good. No similar formula exists for $\mE\aTres$ and $\mE\zTres$ in the situation of random masses. Recall that in our previous paper~\cite{A}, we found that for spatial systems of particles ($d=3$), $\mE\aTres=\mE\zTres\approx 5/(12\nu)$ for sufficiently large $N$ in the situation of equal masses. It would be interesting to figure out how the coefficients $5/16$ and $5/12$ can be generalized to larger dimensions $d$. For instance, the conjectural formula
\[
\mE\aTres=\mE\zTres\approx\frac{5(d-1)}{8d\nu}
\]
gives the desired values $\frac{5}{16\nu}$ and $\frac{5}{12\nu}$ for $d=2$ and $d=3$, respectively, and agrees with the fact that $\Tres\equiv 0$ for $d=1$ (independently of the masses of the particles, see Example~\ref{distwo}).

There are no analogs of the formulas \mbox{\eqref{odin2}--\eqref{pyat2}} for the situation of random masses (except for the equalities $\mE\Tres=\mE\EoutB=0$, of course). However, for large $N$ (starting with $N\approx 50$, say), the following \emph{approximate} equalities hold in our simulations in the situation of random masses:
\begin{gather*}
\overline{\TL}\approx\frac{1.88N-4.04}{2\nu}, \qquad \overline{\Trho}\approx\frac{0.12N+2.04}{2\nu}, \\
\overline{\Trot}\approx\frac{1.83N-5.2}{2\nu}, \qquad \overline{\TI}\approx\frac{0.17N+3.2}{2\nu}, \qquad \overline{\Txi}\approx\frac{0.05N+1.15}{2\nu}, \\
\overline{\Text}=\overline{\TJ}\approx\frac{0.12N+2.05}{2\nu}, \qquad \overline{\Tint}\approx\frac{1.71N-7.22}{2\nu}, \\
\overline{\aTres}=\overline{\zTres}\approx\frac{0.03N+0.79}{2\nu}, \\
\overline{\TK}\approx\frac{0.88N-3.03}{2\nu}, \qquad \overline{\Tac}\approx\frac{0.83N-4.22}{2\nu}, \\
\overline{\EinB}\approx\frac{1.66N-8.39}{2\nu}.
\end{gather*}
The coefficients $a$ and $b$ in these expressions $\overH\approx (aN+b)/(2\nu)$ are obtained for each mean energy $\overH$ via minimizing the sum
\[
\sum_{N=50}^{100} \bigl[ 2(N-1)\overH(N)-aN-b \bigr]^2.
\]
In fact, we found $b=0.78$ for $\overline{\aTres}$ and $b=0.8$ for $\overline{\zTres}$.

Since $\TJ=\Text$ for any $N$ (provided that $d=2$) and $\TK=\Tint$ for $N=3$, one concludes that $\Tac=\Tres$ for $N=3$ independently of the masses. Accordingly, in our simulations with $N=3$, $\overline{\Tac}=0$ and $\Tac$ is negative for half of the systems up to statistical errors for both the mass distributions. For larger values of $N$, one finds $\overline{\Tac}>0$, while the fraction of systems for which $\Tac$ is negative decreases rapidly as $N$ grows for both the mass distributions (see Fig.~\ref{fig3}). However, for each $N\geq 4$, this fraction is greater in the situation of random masses. For particles with random masses, we encountered systems with a negative energy $\Tac$ for any value of $N$. On the other hand, for particles with equal masses for $N=32$, $N=33$, and $37\leq N\leq 100$, the angular coupling energy $\Tac$ of \emph{all} the $10^6$ systems we sampled turned out to be positive (for each value $N=34$, $35$, and $36$, we observed only one system with $\Tac<0$). In fact, systems with the center-of-mass at the origin and with a negative energy $\Tac$ exist for any $d\geq 2$, $N\geq 3$ and any masses $m_1,m_2,\ldots,m_N$, but the ``relative measure'' of the set of such systems is very small for large $N$ and $m_1=m_2=\cdots=m_N$. To ``construct'' a system of $N>3$ particles with prescribed masses and with $\Tac<0$, one may choose a pair $(Z,\dZ)$ corresponding to $N=3$ and $\Tac<0$ and then augment both the matrices $Z$ and $\dZ$ on the right by $N-3$ zero columns (cf.~\cite{A}).

For both the mass distributions and any value of $N$, the fraction of systems for which $\Ec>0$ is much less than $1/2$ and decreases fast as $N$ grows for both the mass distributions (see Fig.~\ref{fig4}). For each $N$, this fraction is greater in the situation of random masses. Nevertheless, even for $N=100$ in the situation with equal masses, the tangent coupling energy $\Ec$ was positive for $32\,969$ systems of $10^6$.

In our numerical experiments for spatial systems of particles~\cite{A}, the energies $\Tac$ and $\Ec$ exhibited a similar behavior.

In a complete analogy with the case $d=3$~\cite{A}, the $14$ energy terms $\TL$, $\Trho$, $\Trot$, $\TI$, $\Txi$, $\Text$, $\Tint$, $\aTres$, $\TK$, $\Tac$, $\aTac$, $\zTac$, $\EinB$, and $\aEc$ in the present simulations for $d=2$ can be divided into two classes, cf.\ Figs.~\mbox{\ref{fig1}--\ref{fig4}}:

\smallskip

A) The seven terms $\TL$, $\Trot$, $\Tint$, $\TK$, $\Tac$, $\aTac$, $\EinB$. For each of these energies $H$,
\begin{equation}
\begin{gathered}
\text{the mean value $\overH(N)$ \emph{increases} as $N$ grows for both the mass distributions, and} \\
\overH(N)\equal > \overH(N)\random, \qquad 3\leq N\leq 100.
\end{gathered}
\label{Abigail}
\end{equation}

B) The seven terms $\Trho$, $\TI$, $\Txi$, $\Text$, $\aTres$, $\zTac$, $\aEc$. For each of these energies $H$,
\begin{equation}
\begin{gathered}
\text{the mean value $\overH(N)$ \emph{decreases} as $N$ grows for both the mass distributions, and} \\
\overH(N)\equal < \overH(N)\random, \qquad 3\leq N\leq 100.
\end{gathered}
\label{Brittany}
\end{equation}

We do not consider here the terms $\Tres$ and $\zTres$ because for both the mass distributions, $\overline{\Tres}=0$ and $\overline{\zTres}=\overline{\aTres}$ for any $N$ up to statistical errors. We do not consider the term $\TJ$ either since $\TJ\equiv\Text$ for $d=2$.

In some cases, the increase/decrease of $\overH(N)$ as $N$ grows is slightly non-monotonous for large $N$ due to statistical errors. Moreover, in the situation of equal masses, as was pointed out above, $\zTac$ turned out to be zero for all the systems sampled for $N\geq 37$ (in fact, $\overline{\zTac}$ was found to be very small for $N\gtrsim 15$).

The exceptions to the rules \mbox{\eqref{Abigail}--\eqref{Brittany}} are as follows. First, $\EinB\equiv 0$ for $N=3$ independently of the masses. Second,
\begin{equation}
\begin{aligned}
& \overline{\Trot}(3)\equal = \overline{\Trot}(3)\random = \tfrac{1}{2} \qquad \text{and accordingly} \\
& \overline{\TI}(3)\equal = \overline{\TI}(3)\random = \tfrac{1}{2}
\end{aligned}
\label{confirm}
\end{equation}
up to statistical errors. Third,
\[
\overline{\Tac}(3)\equal = \overline{\Tac}(3)\random = 0
\]
up to statistical errors. Fourth,
\[
\overline{\Txi}(N)\equal > \overline{\Txi}(N)\random \qquad \text{for} \quad 3\leq N\leq 5
\]
and
\[
\overline{\Txi}(6)\equal \approx \overline{\Txi}(6)\random
\]
within statistical errors. Fifth,
\[
\overline{\aTres}(N)\equal > \overline{\aTres}(N)\random \qquad \text{for} \quad 3\leq N\leq 5.
\]
Sixth,
\begin{align*}
& \overline{\zTac}(N)\equal > \overline{\zTac}(N)\random \qquad \text{and} \\
& \overline{\aEc}(N)\equal > \overline{\aEc}(N)\random \qquad \text{for} \quad 3\leq N\leq 4.
\end{align*}
In our previous paper~\cite{A}, being based on numerical simulations for spatial systems of particles ($d=3$) at $3\leq N\leq 4$ and on the formulas \mbox{\eqref{zwei}--\eqref{drei}} concerning two-particle systems for any $d$, we conjectured that $\mE\Trot=(d-1)/d$ and $\mE\TI=1/d$ at all the values $N\leq d+1$ (and $T=1$) for any mass distribution. The relations~\eqref{confirm} confirm this conjecture for $d=2$, $N=3$.

The larger is the number $N$ of particles, the greater is the contribution of ``kinematic rotations'' to the total kinetic energy $T$ of the system (see Remark~\ref{smysl}) and the smaller is the contribution of conventional rotations and changes in the singular values $\xi_1,\xi_2,\ldots,\xi_\omega$. That is why the mean values of the energies $\TL$, $\Trot$, $\Tint$ (which include the contribution of ``kinematic rotations''), and $\TK$ (which is connected with the kinematic angular momentum $K$) increase as $N$ grows for both the mass distributions whereas the mean values of the energies $\Trho$, $\TI$, $\Txi$, and $\Text=\TJ$ decrease. It is however not clear why the mean values of the terms $\Tac$, $\aTac$, and $\EinB$ increase as $N$ grows while the mean values of the terms $\aTres$, $\zTac$, and $\aEc$ decrease. Why
\begin{equation}
\overH(N)\equal > \overH(N)\random \qquad \text{for increasing terms} \;\;\; \overH(N)
\label{AA}
\end{equation}
and
\begin{equation}
\overH(N)\equal < \overH(N)\random \qquad \text{for decreasing terms} \;\;\; \overH(N)
\label{BB}
\end{equation}
is a complete mystery.

Similarly to the case $d=3$~\cite{A}, the sample variances $s^2(H)$ of each of the $15$ energy terms $\TL$, $\Trho$, $\Trot$, $\TI$, $\Txi$, $\Text$, $\Tint$, $\Tres$, $\aTres$, $\TK$, $\Tac$, $\aTac$, $\zTac$, $\EinB$, and $\aEc$ for $d=2$ \emph{decrease} as the number $N$ of particles grows for both the mass distributions [recall that $s^2(\zTres)=s^2(\aTres)$ up to statistical errors]. For large $N$, this decrease is sometimes slightly non-monotonous due to statistical errors. Of course, $s^2(\zTac)=0$ in the situation of equal masses for $N\geq 37$. As one expects, the sample variance of each of these terms in the situation of random masses is \emph{larger} than that in the situation of equal masses for the same value of $N$. For all the terms $H$, the sample variance $s^2(H)$ gets very small in the situation of equal masses for large $N$.

The exceptions to these rules are as follows. First, in the situation of random masses, $s^2(\Tint)$ decreases starting with $N=4$ (rather than with $N=3$): this quantity for $N=4$ is larger than for $N=3$. Second, for both the mass distributions, $s^2(\Tac)$ decreases starting with $N=4$. Third, $s^2(\aTac)$ decreases starting with $N=4$ in the situation of equal masses and with $N=5$ in the situation of random masses. Fourth, $s^2(\EinB)$ decreases starting with $N=4$ in the situation of equal masses and with $N=6$ in the situation of random masses. Moreover, the variance of $\EinB$ for $N=3$ is of course zero independently of the masses. Apart from this, the inequality
\begin{equation}
s^2(H)\equal < s^2(H)\random
\label{CC}
\end{equation}
is violated in the following cases:
\begin{align*}
& \text{for} \quad H=\Tint \text{ and } \TK & & \text{at} \quad N=3, \\
& \text{for} \quad H=\EinB & & \text{at} \quad N=4, \\
& \text{for} \quad H=\zTac \text{ and } \aEc & & \text{at} \quad 3\leq N\leq 4, \\
& \text{for} \quad H=\Txi, \; \Tres, \; \aTres, \; \Tac, \text{ and } \aTac & & \text{at} \quad 3\leq N\leq 5.
\end{align*}
As one sees, there is a strong correlation between the violation of the inequality~\eqref{CC} and that of the inequalities \mbox{\eqref{AA}--\eqref{BB}}.

Our simulations confirm that the projective partition~\eqref{projective} ensures a very effective separation between the conventional rotations and ``kinematic rotations'' \cite{A2,A,A-Appl1,A-Appl2} compared with the hyperspherical partition~\eqref{hyperspherical}, not to mention the singular value expansion~\eqref{singular}. The mean absolute value $\overline{|\Tres|}$ of the residual energy decreases as $N$ grows for both the mass distributions, whereas the mean absolute value $\overline{|\Tac|}$ of the angular coupling energy increases. For $N=3$ one has $\Tres=\Tac$ independently of the masses, and consequently $\overline{|\Tres|}(3)=\overline{|\Tac|}(3)$; these mean values are equal to $0.168796$ in the situation of equal masses and to $0.146767$ in the situation of random masses. For $N\geq 4$, one has $\overline{|\Tres|}<\overline{|\Tac|}$, and the larger $N$, the greater is this difference. At $N=4$
\begin{align*}
& \overline{|\Tres|}(4)\random=0.136738 < \overline{|\Tres|}(4)\equal=0.151223 < {} \\
& \overline{|\Tac|}(4)\random=0.214097 < \overline{|\Tac|}(4)\equal=0.250362,
\end{align*}
while at $N=100$
\begin{align*}
& \overline{|\Tres|}(100)\equal=0.00640506 < \overline{|\Tres|}(100)\random=0.0381073 \ll {} \\
& \overline{|\Tac|}(100)\random=0.400477 < \overline{|\Tac|}(100)\equal=0.489895.
\end{align*}

\section{Conclusions}\label{conclusion}

The statistical studies of our previous paper~\cite{A} (devoted to systems in $\mR^3$) and those of the present paper (devoted to systems on $\mR^2$) are formal in the sense that they do not take into account any interaction potentials between the particles. If one considers kinetic energy partitions for interacting particles with a certain potential energy $\eU$, then it is more natural to average various energy terms $H$ at a fixed total energy $T+\eU$ (averaging over a \emph{microcanonical ensemble}, see e.g.~\cite{Ruelle}) rather than at a fixed total kinetic energy $T$, cf.\ \cite{Marsden1,Marsden2,A2,A-Appl1,A-Appl2}. Choosing potential energy hypersurfaces at random according to some distribution in an appropriate infinite dimensional functional space, one would probably obtain entirely different statistics of the energy components. In this case, it seems suitable to average the ratios $H/T$ or $H/(T+\eU)$ over the initial conditions, the potential $\eU$, and the time. It is also of interest to compute the mean values of the energy terms at fixed values of the total angular momentum $J$ (which is, by the way, customary in quantum mechanics, see e.g.~\cite{Zhang}) or kinematic angular momentum $K$.

There are many ways to generalize the energy partitions treated in the present work. One of them is pointed out in our previous paper~\cite{A} and consists in defining the energy terms corresponding to the actions of arbitrary subgroups of the orthogonal groups $\rO(d)$ and $\rO(N)$. Another approach recently proposed by Marsden and coworkers~\cite{Marsden2} for $d=3$ is called the \emph{hyperspherical mode analysis} by the authors. The $3N-6$ internal modes of an $N$-atom system ($N\geq 5$) in $\mR^3$ are classified in~\cite{Marsden2} into three \emph{gyration-radius} modes, three \emph{twisting} modes, and $3N-12$ \emph{shearing} modes. Most probably, Marsden's theory can be generalized to the case of an arbitrary dimension $d$.

One of the main results of our previous paper~\cite{A} and the present paper is that in the situation of equal masses, the mean values $\mE H$ of various components $H$ of the total kinetic energy $T$ are expressed in terms of the dimension $d$ of the physical space and the number $N$ of particles in a very simple way. However, it is not clear at all whether the \emph{distributions} of $H$ are ``simple'' functions of $d$, $N$, and $H$, not to mention the joint distributions of several components. For instance, we have not attempted to find any expressions for the variances $\Var H$ of $H$ [or, equivalently, for $\mE(H^2)$] or, say, for the correlation coefficients \cite{Cramer,Lagutin,Waerden,Wilks} between the energy terms. We hope that such detailed statistical properties of the kinetic energy partitions of classical systems will be examined (both numerically and rigorously) in further research.

The work of MBS was supported in part by a grant of the President of the Russia Federation, project No.\ NSh-4850.2012.1.

\end{document}